\documentclass[10pt]{IEEEtran}

\usepackage{amsmath}
\usepackage{mathtools}
\usepackage{graphicx}
\usepackage{enumerate}
\usepackage{url}
\usepackage{subfigure}
\usepackage{algorithm}
\usepackage{algpseudocode}
\usepackage{stfloats}
\usepackage{lineno}
\usepackage{hyperref}
\usepackage{bm}
\usepackage{bbm}
\usepackage{amssymb}
\usepackage{enumerate}
\usepackage{xcolor}
\usepackage{float}
\usepackage{cite}
\usepackage{tabularx}
\usepackage{tabu}
\usepackage{multicol}
\usepackage{multirow}
\usepackage{colortbl,booktabs,threeparttable}
\usepackage{dcolumn}
\usepackage{flushend}
\usepackage{soul}
\usepackage{CJKutf8}
\usepackage{ragged2e}

\graphicspath{{Figures/}}

\newcommand{\rscl}[1]{\mathrm{#1}}  
\renewcommand{\vec}[1]{\bm #1}
\newcommand{\rvec}[1]{\mathbf{#1}}
\newcommand{\mat}[1]{\bm #1}
\newcommand{\rmat}[1]{\mathbf{#1}}
\newcommand{\vech}[1]{\hat{\bm #1}}

\renewcommand{\math}[1]{\hat{\bm #1}}

\newcommand{\matb}[1]{\bar{\bm #1}}

\renewcommand{\cal}[1]{\mathcal{#1}}
\newcommand{\bfit}[1]{\textbf{\textit{#1}}}

\newcommand{\Tr}{\operatorname{Tr}}

\newcommand{\C}{\mathbb{C}}
\newcommand{\E}{\mathbb{E}}

\renewcommand{\P}{\mathbb{P}}
\newcommand{\Q}{\mathbb{Q}}
\newcommand{\Ph}{\hat{\mathbb{P}}}

\newcommand{\Po}{\mathbb{P}_0}
\newcommand{\Pb}{\bar{\mathbb{P}}}

\renewcommand{\d}{\mathrm{d}}
\renewcommand{\th}{\text{th}}

\newcommand{\T}{\mathsf{T}}
\renewcommand{\H}{\mathsf{H}}
\renewcommand{\st}{\text{s.t.}}

\DeclareMathOperator*{\argmin}{argmin}

\newcommand{\Px}{\P_{\rvec x}}

\newcommand{\Ripn}{\mat {R}_{{i+n}}}
\newcommand{\Ripnh}{\math{R}_{{i+n}}}

\newcommand{\defeq}{\coloneqq}
\newcommand{\stp}{\hfill $\square$}     
\newcommand{\captext}[1]{\texorpdfstring{#1}{}} 

\newcommand{\quotemark}[1]{``#1”}

\definecolor{hl-bg-color}{RGB}{255,255,215}
\sethlcolor{hl-bg-color} 
\definecolor{new-magenta}{RGB}{255,0,255}
\soulregister\cite7
\soulregister\citep7
\soulregister\citet7
\soulregister\ref7
\soulregister\eqref7
\newcommand*{\HIGHLIGHT}{}
\ifdefined\HIGHLIGHT

\else

\fi

\newtheorem{theorem}{{Theorem}}

\newtheorem{corollary}{{Corollary}}

\newtheorem{method}{{Method}}

\newtheorem{remark}{{Remark}}

\newtheorem{definition}{{Definition}}

\newtheorem{example}{{Example}}

\newtheorem{insight}{{Insight}}

\begin{document}
\newpage
\title{
Distributionally Robust Adaptive Beamforming
}

\author{Shixiong Wang,~
        Wei~Dai,~
        Geoffrey~Ye~Li,~\IEEEmembership{Fellow,~IEEE}
\thanks{The authors are with the Department of Electrical and Electronic Engineering, Imperial College London, London SW7 2AZ, United Kingdom (E-mails: s.wang@u.nus.edu; wei.dai1@imperial.ac.uk; geoffrey.li@imperial.ac.uk).
}
}

\maketitle

\begin{abstract}
As a fundamental technique in array signal processing, beamforming plays a crucial role in amplifying signals of interest (SoI) while mitigating interference plus noise (IPN). 
When uncertainties exist in the signal model or the data size of snapshots is limited, the performance of beamformers significantly degrades. 
In this article, we comprehensively study the conceptual system, theoretical analysis, and algorithmic design for robust beamforming against uncertainties in the assumed snapshot or IPN covariances. Since such robustness is specific to the probabilistic uncertainties of snapshots or IPN signals, it is referred to as distributional robustness. Particularly, four technical approaches for distributionally robust beamforming are proposed, including locally distributionally robust beamforming, globally distributionally robust beamforming, regularized beamforming, and Bayesian-nonparametric beamforming. 
In addition, we investigate the equivalence among the four technical approaches and suggest a unified distributionally robust beamforming framework. Moreover, we show that the resolution of power spectra estimation using distributionally robust beamforming can be greatly refined by incorporating the characteristics of subspace methods, and hence, the accuracy of IPN covariance reconstruction can be improved, especially when the interferers are close to the SoI. As a result, the robustness of beamformers based on IPN covariance estimation can be further enhanced.
\end{abstract}


\section{Introduction} \label{sec:introdction}
\IEEEPARstart{A}{\lowercase{daptive}} beamforming, or simply beamforming, has shown numerous successful applications in array signal processing, e.g., wireless communications, radar, and sonar, to enhance signals of interest (SoI) and suppress interference plus noise (IPN). Typical examples include the waveform, power, and direction-of-arrival (DoA) estimation of the SoI, as well as the maximization of the array's output signal-to-interference-plus-noise ratio (SINR) \cite{van2002optimum,li2006robust,benesty2017fundamentals,elbir2023twenty}. To achieve these application purposes, a large body of representative beamformers have been put forward \cite{van2002optimum,benesty2017fundamentals}, such as the minimum mean-squared error (MMSE) beamformer for waveform estimation, the Bartlett beamformer (i.e., delay-and-sum beamformer, conventional beamformer, or matched filter) for power and DoA estimation, the minimum power distortionless response (MPDR) beamformer (i.e., Capon beamformer) for power and DoA estimation, and the minimum variance distortionless response (MVDR) beamformer for SINR maximization, etc. The difference between the MPDR and MVDR beamformers is whether the snapshot covariance or the IPN covariance is used \cite{ehrenberg2010sensitivity}, that is, SoI-contaminated or SoI-free \cite{vorobyov2003robust}. In the SoI-free case, the snapshot covariance coincides with the IPN covariance, so the MPDR and MVDR beamformers become identical. As a convention and without loss of generality \cite[p.~1705]{li2003robust}, \cite{lorenz2005robust}, this article employs SINR as the performance metric for beamformer design, and hence, the MPDR and MVDR beamformers are particularly focused on \cite{kim2008robust,huang2023robust}.

\subsection{Uncertainty Issues and Literature Review}\label{subsec:literature-review}
The SINR performance of MPDR and MVDR beamformers, however, degrades significantly when uncertainties exist in the assumed steering vector or in the estimated snapshot and IPN covariances \cite{cox1987robust,guerci1999theory,vorobyov2003robust,vorobyov2013principles}, although in such a case the MVDR beamformer tends to outperform the MPDR beamformer \cite{ehrenberg2010sensitivity}. These uncertainties include, for example, array calibration errors, array pointing (i.e., looking-direction) errors, limited data size of snapshots, and non-stationarity of emitting signals or channel noises (e.g., time-varying powers) \cite{vorobyov2013principles,li2006robust,guerci1999theory}. To handle these uncertainties, several archetypal treatments have been introduced in the past to improve the robustness of beamformers. Herein, the robustness of a beamformer means the ability to be insensitive to reasonably small uncertainties \cite{cox1987robust,vorobyov2003robust,kim2008robust}, and consequently, the performance of a robust beamformer in real-world operation can remain as satisfactory as that in the design phase  \cite{vorobyov2013principles}. As per the design motivations, existing treatments can be categorized into three main streams: 1) robustness against uncertainties in steering vectors, 2) robustness against uncertainties in snapshot covariances, and 3) robustness based on IPN covariance reconstruction.

\subsubsection{Robustness Against Uncertainties in Steering Vectors}
Uncertainties in steering vectors come from array calibration errors, pointing errors (i.e., DoA mismatches), etc \cite{vorobyov2003robust,li2003robust,lorenz2005robust,khabbazibasmenj2012robust,huang2023robust}. Representatives in this stream include the following: 
\begin{enumerate}[\text{1}a)]
\item 
Linearly constrained minimum variance (LCMV) beamforming \cite{steele1983comparison} \cite[Section~2.5]{vorobyov2013principles}, \cite[Section~6]{er1985alternative};

\item Covariance matrix tapering to widen nulls and beams, which leverages the phase dithering effect \cite{guerci1999theory}, \cite[Section~2]{shahbazpanahi2003robust}. Some technical equivalences between this approach and Method 1a) are noted in \cite{zatman2000comments};

\item 
Bayesian beamforming, which assumes a probability distribution of the DoA, and then obtains the expected beamformer under this DoA distribution  \cite{bell2000bayesian,chakrabarty2017bayesian}; 

\item 
Beampattern control, which, conceptually similar to Method 1a), requires the array power response to comply with some prerequisites \cite{er1985alternative,liao2017robust};

\item 
Steering vector estimation via eigenspace projection \cite{chang1992performance,feldman1994projection} or output power maximization \cite{hassanien2008robust,khabbazibasmenj2012robust,huang2019new}. The eigenspace projection indirectly increases the array output power in the MPDR beamforming scheme. Also, the covariance-fitting methods in \cite[Eqs.~(8), (9)]{li2003robust} and \cite[Eqs.~(10), (13)]{stoica2003robust} amount to output power maximization;

\item 
Worst-case array unity response  \cite{vorobyov2003robust,li2003robust,shahbazpanahi2003robust,lorenz2005robust,huang2020quadratic,huang2023robust}. Regarding objective functions, Method 1e) is technically equivalent to this approach under the MPDR beamforming scheme; see \cite[Eq.~(17)]{kim2008robust}, \cite[Eq.~(23)]{khabbazibasmenj2012robust}, \cite[Eq.~(12)]{huang2019new}, \cite[Appendices A,~C]{li2003robust};

\item 
Probabilistic worst-case array unity response, which can be technically cast into Method 1f) \cite{vorobyov2008relationship};

\item
Imposing quadratic constraints on beamformers \cite[p.~1366]{cox1987robust}, \cite[p.~742]{rubsamen2011robust}. Some beampattern control problems are technically equivalent to this approach \cite[Eq.~(22)]{er1985alternative}.
\end{enumerate}
Methods 1a), 1b), and 1c) are ad hoc to DoA mismatches, while Methods 1d), 1e), 1f), 1g), and 1h) can apply to arbitrary types of uncertainties in the steering vector. The philosophy of robustness behind Methods 1a), 1b), 1d), 1f), and 1g) is to identify a plausible uncertainty set for the actual but unknown steering vector, and then guarantee that the array power responses for all elements in this set are controlled, for example, greater than unity. In the case of DoA mismatch, the uncertainty set is induced by the uncertainty region of DoA under the array geometry. In contrast, the philosophy of robustness behind Method 1c) is to consider an uncertainty region of DoA through employing a distribution of DoA. Moreover, the philosophy of robustness behind Method 1e) is that the actual steering vector should lie in the subspace of signal-plus-interference or should increase the array output power. In addition, the philosophy of robustness behind Method 1h) is to limit the array's sensitivity to various errors because the sensitivity can be measured by the 2-norm of beamformers.

All methods, except Method 1e), can be integrated or adapted into both the MPDR and MVDR beamforming to enhance their robustness. Method 1e) cannot be considered in the MVDR beamforming scheme because in this case, the subspace of signal-plus-interference and the array output power are not applicable.

\subsubsection{Robustness Against Uncertainties in Snapshot Covariances}
Uncertainties in snapshot covariances can be caused by the limited data size of snapshots, the non-stationarity of signal characteristics, etc \cite{carlson1988covariance,guerci1999theory,shahbazpanahi2003robust,stoica2008using,yang2018high,huang2023robust}. 
When the SoI is absent in snapshots, the snapshot covariance coincides with the IPN covariance. 
Representatives in this stream include the following:
\begin{enumerate}[\text{2}a)]
    \item Advanced techniques to estimate the covariance matrix using snapshots, for example, the M-estimator \cite{elkhalil2017fluctuations,zoubir2012robust}, the spiked covariance estimator \cite{yang2018high}, the eigenvalue thresholding method \cite{harmanci2000relationships}, the diagonal loading method \cite{carlson1988covariance,mestre2005finite,zhang2016robust,de2018loading}, the prior-knowledge embedding method \cite{stoica2008using}, etc. The M-estimator can suppress outliers in snapshots, while the remaining can combat snapshot scarcity.

    \item Worst-case optimization \cite[Eq.~(46)]{shahbazpanahi2003robust}, \cite[Eq.~(10)]{huang2023robust}, \cite[Eq.~(9)]{kim2008robust}. This approach can be shown to technically amount to the diagonal loading method \cite{shahbazpanahi2003robust}.
\end{enumerate}
The philosophy of robustness behind Method 2a) is to obtain a better covariance estimate, which is closer to the ground truth than the usual sample covariance matrix. By feeding this better covariance to beamforming, the performance is expected to improve. On the contrary, the philosophy of robustness behind Method 2b) is to optimize the worst-case performance, e.g., to maximize the worst-case SINR \cite{kim2008robust}. In \cite{shahbazpanahi2003robust,huang2023robust}, this worst-case SINR is determined by minimizing SINR over the covariance estimation error, which is characterized by an F-norm ball.

Depending on whether the SoI is present in snapshots or not, these methods correspond to the robust MPDR and MVDR beamformers, respectively.

\subsubsection{Robustness Based on IPN Covariance Reconstruction}
Motivated by the fact that the MVDR beamformer tends to be more robust than the MPDR beamformer \cite{ehrenberg2010sensitivity}, another philosophy to achieve robustness is to estimate the IPN covariance matrix from snapshots before beamforming. This approach applies to the case where the SoI is present in snapshots. Representatives in this stream include the following:
\begin{enumerate}[\text{3}a)]
\item Strategies via spectra estimation, for example, the Capon spectra \cite{gu2012robust} and the maximum-entropy spectra \cite{mohammadzadeh2020maximum}. Improvements in terms of uncertainty-awareness and computation reduction are reported in \cite{huang2015robust} and \cite{zhang2015interference,mohammadzadeh2022covariance}, respectively;

\item Strategies via SoI cancellation, e.g., \cite[Eq.~(26)]{ruan2013robust} and \cite[Eq.~(27)]{ruan2016robust}, where the snapshot covariance is first estimated and then the IPN covariance is recovered through subtracting the SoI component from the snapshot covariance.
\end{enumerate}

\subsection{Problem Statements}\label{subsec:problem-statements}
Although various challenges in robust adaptive beamforming have been attacked, the following issues remain unsolved.
\begin{enumerate}[\text{I}1)]
    \item In the literature, although the notion of \quotemark{robust beamforming} is widely used, the quantification and formalization have never been studied. The only effort in this direction is found in, e.g., \cite{cox1987robust,rubsamen2011robust}, where the authors propose to use the white-noise gain \cite{cox1987robust} or the (normalized) squared magnitude \cite{rubsamen2011robust} as a quantitative measure of the robustness of a beamformer. However, this robustness formalism does not apply to the uncertainties in the snapshot and IPN covariances. On the other hand, although numerous minimax-optimization-based approaches claim the robustness of their resultant beamformers \cite[Eq.~(46)]{shahbazpanahi2003robust}, \cite[Eq.~(10)]{huang2023robust}, the rigorous relation between the beamformer's robustness and minimax optimization remains unclear. Therefore, a comprehensive study of the conceptual system, theoretical analysis, and algorithmic design for robust beamforming, \textit{specifically against uncertainties in snapshot and IPN covariances}, has to be conducted.

    \item Regarding robustness against uncertainties in snapshot covariances, the relations between Methods 2a) and 2b) are unclear. To clarify further, in the existing literature, researchers tacitly employ the eigenvalue thresholding method, the diagonal loading method, or the prior-knowledge embedding method, and then focus on tuning the involved parameters using different empirical or algorithmic criteria \cite{carlson1988covariance,mestre2005finite,stoica2008using,zhang2016robust,de2018loading}. However, the rationale behind these choices has not been systematically elaborated. A specific question is as follows: In addition to diagonal loading \cite{shahbazpanahi2003robust}, can eigenvalue thresholding \cite{harmanci2000relationships} and prior-knowledge embedding \cite{stoica2008using} also be shown to optimize the worst-case performance? This issue remains valid for both the MPDR beamformer that relies on the snapshot covariance and the MVDR beamformer that relies on the IPN covariance.

    \item In the robustness framework based on IPN covariance reconstruction, i.e., Methods 3a) and 3b), the estimation error of IPN is also unavoidable. Hence, it is natural to ask: Can the robustness of this framework be further improved using the robustness strategies in Methods 2a) and 2b)? This question has not been answered in the literature.

    \item In Method 3a) for IPN covariance reconstruction, the resolution of spectra estimation should be sufficiently high to suppress interferers that locate close to the SoI. This concern has also been raised in \cite[p.~1647]{huang2015robust}, \cite[Fig.~1]{mohammadzadeh2020maximum}. However, the resolution of the Capon spectra and the maximum-entropy spectra are still limited \cite[Fig.~3]{schmidt1986multiple}, \cite{johnson1982application}, and therefore, higher-resolution spectra estimation methods are expected to improve the estimation accuracy of the IPN covariance, especially when interferers are close to the SoI. The challenge is that modern high-resolution spectra estimation methods, such as multiple signal classification (MUSIC) and other subspace strategies, do not directly provide spectral estimates. Instead, they generate pseudo-spectra to identify the DoAs of interferers. On the other hand, to ensure high accuracy of IPN covariance reconstruction, the power spectra estimation should be unbiased, that is, no gaps between identified DoAs and true DoAs of interferers. However, the maximum-entropy spectra tend to have large biases \cite[Fig.~4]{schmidt1986multiple}. As for Method 3b), it innately assumes that interferers are well-separated from the SoI, which limits its applicability when interferers are close to the SoI.
\end{enumerate}

\subsection{Contributions}
To address the four issues aforementioned, this article makes the following contributions.
\begin{enumerate}[\text{C}1)]
    \item We propose a quantitative definition of robustness \textit{against uncertainties in snapshot and IPN covariances}, and prove that many existing beamformers in Methods 2a) and 2b) are indeed robust; for example, diagonal loading, eigenvalue thresholding, prior-knowledge embedding can be shown to amount to worst-case optimization; see Definitions \ref{def:local-robust-beamformer} and \ref{def:global-robust-beamformer}, Theorems \ref{thm:robust-bf-R0}, \ref{thm:sol-minimax-robust-bf-Rh}, and \ref{thm:sol-robust-bf-Rh-global}, Examples \ref{ex:eig-thres}, \ref{ex:diagonal-loading}, \ref{ex:regularization}, \ref{ex:prior-knowledge}, and \ref{ex:diag-loading-bayes}, Corollaries \ref{cor:sol-capon-robust-R} and \ref{cor:sol-capon-robust-R-2}, and Insights \ref{insight:regularization} and \ref{insight:bayes-model}. To emphasize the probabilistic nature of snapshots and IPN signals, and differentiate from the robustness against uncertainties in steering vectors, we call the introduced concept as \quotemark{\textit{distributional robustness}}. Subsequently, four technical approaches for distributionally robust beamforming are proposed, i.e., locally distributionally robust beamforming, globally distributionally robust beamforming, regularized beamforming, and Bayesian-nonparametric beamforming; see Models \eqref{eq:robust-bf-Rh}, \eqref{eq:robust-bf-Rh-global}, \eqref{eq:sol-capon-robust-R-2}, and \eqref{eq:capon-bayesian-compact}. The equivalence among these approaches is investigated and a unified distributionally robust beamforming framework is suggested; see Insight \ref{insight:unified-framework}. 

    \item We show that by incorporating the characteristics of the MUSIC method into the proposed distributionally robust (DR) beamforming framework, the resolution of spectra estimation can be greatly refined. As a result, the SINR performance of DR beamformers based on IPN covariance reconstruction can be largely improved when interferers are close to the SoI; see Insight \ref{insight:high-resolution}, Method \ref{method:robust-capon-beamforming}, and Algorithm \ref{algo:bf-doa}.
\end{enumerate}

Contribution C1) solves Issues I1), I2), and I3), while Contribution C2) addresses Issue I4).

\subsection{Notations} 
Uppercase symbols (e.g., $\mat X$) denote matrices while lowercase ones are reserved for vectors (e.g., $\vec x$). We use upright and italic fonts for random and deterministic quantities, respectively; e.g., $\rmat X$ and $\rvec x$, and $\mat X$ and $\vec x$. Let $\C^d$ denote the $d$-dimensional space of complex numbers. The running index set $[K]$ induced by integer $K$ is defined as $[K] \defeq \{1, 2, 3, \ldots, K\}$. Let $\|\mat X\|$, $\Tr \mat X$, $\mat X^{-1}$, $\mat X^\T$, and $\mat X^\H$ denote a norm, the trace, the inverse (if exists), the transpose, and the conjugate transpose of matrix $\mat X$; the definition of a matrix norm will be specified in contexts. For two matrices $\mat A$ and $\mat B$, $\mat A \succeq \mat B$ means that $\mat A - \mat B$ is positive semidefinite. The $d$-dimensional identity matrix is written as $\mat I_d$ and a zero matrix/vector with compatible dimensions is as $\vec 0$. Let $\cal{CN} (\vec \mu, \mat \Sigma, \mat \Sigma')$ denote the complex normal distribution with mean $\vec \mu$, covariance $\vec \Sigma$, and pseudo-covariance $\mat \Sigma'$; if $\mat \Sigma'$ is not specified, we admit $\mat \Sigma' = \mat 0$. Let $\E_{\P}[\cdot]$ denote the expectation operator under distribution $\P$.



\section{System Model, Existing Works, and Research Questions}\label{sec:problem-formulation}
\subsection{System Model}\label{subsec:system-model}
This article focuses on base-band and narrow-band signal processing. Suppose $K$ signals impinge on an $N$-element antenna array. Let $\rvec x \in \C^N$ denote the array input (i.e., snapshot), $\rscl s_k \in \C$ the $k^\th$ incident signal for $k \in [K]$, and $\theta_k$ the DoA of $\rscl s_k$. The signal model for a single snapshot is
\begin{equation}\label{eq:signal-model}
    \rvec x = \sum^K_{k=1} \vec a_0(\theta_k) \rscl s_k + \rvec v,
\end{equation}
where $\vec a_0(\theta) \in \C^N$ denotes the array steering vector in direction $\theta$ and $\rvec v \in \C^N$ the channel noise. In usual array signal processing literature, it is assumed that $\rscl s_k \sim \cal{CN} (0, \sigma^2_k)$ and $\rvec v \sim \cal{CN} (\bm 0, \sigma^2_n \mat I_N)$, in which $\sigma^2_k$ is the signal power and $\sigma^2_n$ is the noise power. Without loss of generality, we suppose that $\rscl s_1$ is the signal of interest and $\rscl s_k$ for $k = 2, 3, \ldots, K$ are interference signals. Let $\vec w \in \C^N$ be a beamformer. The corresponding array output is 
\begin{equation}\label{eq:array-output}
    \rscl y = \vec w^\H \rvec x = \vec w^\H \vec a_0(\theta_1) \rscl s_1 + \vec w^\H \sum^K_{k=2} \vec a_0(\theta_k) \rscl s_k + \vec w^\H \rvec v,
\end{equation}
and the array output SINR is
\begin{equation}\label{eq:sinr}
    h(\vec w, \vec a_0(\theta_1), \Ripn) \defeq \frac{\sigma^2_1 \vec w^\H \vec a_0(\theta_1) \vec a_0(\theta_1)^\H \vec w}{\vec w^\H \Ripn \vec w},
\end{equation}
where 
$
    \Ripn \defeq \sum^K_{k=2} \sigma^2_k \vec a_0(\theta_k) \vec a_0(\theta_k)^\H + \sigma^2_n \mat I_N
$ 
denotes the covariance of the interference signals plus noise. Depending on specific applications, the typical roles of beamforming include the following: to adjust array output SINR \eqref{eq:sinr}, to estimate SoI waveform $\rscl s_1$ using $\rscl y$ \cite{elbir2023twenty,wang2025distributionally}, to estimate SoI power $\sigma^2_1$ using $\E \rscl y \rscl y^\H$ \cite{li2003robust,stoica2003robust}, and to estimate DoAs $\theta_k$ using the directions corresponding to the largest output powers $\E \rscl y \rscl y^\H$ \cite[p.~1019]{johnson1982application}.

The MVDR beamformer in direction $\theta$ solves the following beamforming problem: 
\begin{equation}\label{eq:mvdr}
    \begin{array}{cl}
      \displaystyle \min_{\vec w}  &  \vec w^\H \Ripn \vec w \\
      \st  &  \vec w^\H \vec a_0(\theta) = 1.
    \end{array}
\end{equation}
The MVDR beamformer is optimal in many senses: to achieve maximum output SINR \eqref{eq:sinr} and to attain maximum likelihood estimate of SoI $\rscl s_1$ given $\rvec x$ when $\theta \defeq \theta_1$ \cite[Eq.~(20)]{johnson1982application}. In practice, however, the IPN covariance $\Ripn$ is unknown, and alternatively, the following MPDR (i.e., Capon) beamforming problem that minimizes the array output power $\E \rscl y \rscl y^\H$ in direction $\theta$ is solved:
\begin{equation}\label{eq:capon}
    \begin{array}{cl}
      \displaystyle \min_{\vec w}  &  \vec w^\H \mat R_x \vec w \\
      \st  &  \vec w^\H \vec a_0(\theta) = 1,
    \end{array}
\end{equation}
where $\mat R_x \defeq \E_{\rvec x \sim \Px}[\rvec x \rvec x^\H]$ denotes the covariance of received signal $\rvec x$ and $\Px$ its underlying true distribution; note that $\mat R_x$ can be estimated using collected snapshots. In what follows, we use $\vec a_0$ as a shorthand for $\vec a_0(\theta)$, if no ambiguity is caused. Among existing beamformers, when the snapshot covariance $\mat R_x$ and the array steering vector $\vec a_0$ are (almost) exactly known, the MPDR beamformer is most popular for its excellent performance and real-world operationality. To be specific, the MPDR beamformer is optimal in the sense of MPDR waveform estimation \cite{lorenz2005robust}, \cite[Eq.~(4)]{elbir2023twenty}, MPDR power estimation \cite[Eq.~(7)]{li2003robust}, high-resolution DoA estimation \cite{capon1969high,johnson1982application}, and SINR maximization \cite[p.~1540]{kim2008robust}, \cite[Eq.~(10)]{vorobyov2013principles}. In addition, the snapshot covariance can be estimated using the array input data, while the IPN covariance in the MVDP beamforming is difficult to obtain for the SoI-contaminated case. However, if the SoI $\rscl s_1$ is absent in the snapshot $\rvec x$, the snapshot covariance $\mat R_x$ matches the IPN covariance $\Ripn$, and the MPDR beamformer becomes the MVDR beamformer. 

Hereafter, according to whether the MPDR or MVDR beamforming scheme is used, we let $\mat R_0$ denote either the snapshot covariance $\mat R_x$ or the IPN covariance $\Ripn$. 
Let $\math R$, $\math R_x$, $\Ripnh$, and $\vech a$ denote the estimates of $\mat R_0$, $\mat R_x$, $\Ripn$, and $\vec a_0$, respectively.

\subsection{Existing Works}\label{subsec:existing-works}
To address the uncertainty in $\math R$ compared to its true value $\mat R_0$, typical treatments include 1) eigenvalue thresholding, 2) diagonal loading, 3) prior-knowledge embedding, 4) worst-case optimization, and 5) IPN covariance reconstruction.

\textit{Eigenvalue Thresholding}: Let $(\lambda_1, \lambda_2, \ldots, \lambda_N)$ be eigenvalues of $\math R$ in descending order and $\math U$ contains eigenvectors. Define $\math R_{\text{thr}} \defeq
\math U \mat \Lambda_{\text{thr}} \math U^\H$ and 
\begin{equation}\label{eq:eig-thres}
     \mat \Lambda_{\text{thr}} \defeq 
\left[ 
\begin{array}{cccc}
    \lambda_1 & & & \\
    & \max\{\mu \lambda_1, \lambda_2\}& & \\
    & & \ddots & \\
    & & & \max\{\mu \lambda_1, \lambda_N\}
\end{array}
\right]
\end{equation}
where $0 \le \mu \le 1$; NB: when $\mu = 0$, $\math R_{\text{thr}}$ reduces to $\math R$. The eigenvalue thresholding method to combat the uncertainty in $\math R$ is to use $\math R_{\text{thr}}$ in beamforming \cite[Eq.~(36)]{harmanci2000relationships}, \cite[Eq.~(12)]{lorenz2005robust}.

\textit{Diagonal Loading}: The diagonal loading method to combat the uncertainty in $\math R$ is to use $\math R + \epsilon \mat I_N$ in beamforming, where $\epsilon \ge 0$ is a scalar \cite{carlson1988covariance,mestre2005finite,zhang2016robust,de2018loading}. Different tuning principles for $\epsilon$ have been discussed, including main beam correction and side lobe reduction \cite{carlson1988covariance,guerci1999theory}, asymptotic properties \cite{mestre2005finite}, empirical mean-squared error (MSE) minimization \cite{stoica2008using}, SoI power estimation \cite{zhang2016robust}, etc. An interesting observation is that some robust methods against the uncertainty in the assumed steering vector $\vech a$ can be shown as the diagonal loading method, where $\epsilon$ is determined by the uncertainty degree of $\vech a$ \cite{vorobyov2003robust,li2003robust} or the sensitivity measure of the array \cite{cox1987robust,rubsamen2011robust}. In practice, however, the best tuning method for $\epsilon$ is still trial-and-error because the best value under one criterion challenges the optimality (or even satisfaction) under the other criterion.

\textit{Prior-Knowledge Embedding}: The prior-knowledge embedding method to combat the uncertainty in $\math R$ is to use $\alpha \math R + \beta \matb R$ in beamforming, for some weight coefficients $\alpha, \beta \ge 0$ and prior knowledge $\matb R$ of $\mat R_0$ \cite{stoica2008using}. It is believed that $\alpha \math R + \beta \matb R$ can provide a better estimate than $\math R$ and $\matb R$, in the sense of smaller MSE. The coefficients $\alpha$ and $\beta$ can be tuned using empirical MSE minimization \cite{stoica2008using}.

\textit{Worst-Case Optimization}:
Robust beamforming in the sense of minimax optimization is formulated as \cite{vorobyov2003robust,shahbazpanahi2003robust,lorenz2005robust,huang2023robust}
\begin{equation}\label{eq:robust-capon}
    \begin{array}{cl}
      \displaystyle \min_{\vec w} \max_{\mat R \in \cal U_R}  &  \vec w^\H \mat R \vec w \\
      \st  & \displaystyle \min_{\vec a \in \cal U_a}  \vec w^\H \vec a \vec a^\H \vec w \ge 1,
    \end{array}
\end{equation}
where $\cal U_R$ and $\cal U_a$ are the uncertainty sets for $\math R$ and $\vech a$, respectively; $\math R, \mat R_0 \in \cal U_R$ and $\vech a, \vec a_0 \in \cal U_a$. The philosophy behind \eqref{eq:robust-capon} is to \textit{optimize the worst-case performance and guarantee the worst-case feasibility}. 
When 
$
\cal U_R \defeq \{\mat R:~\|\mat R - \math R\|_F \le \epsilon_1\}
$ 
and 
$
\cal U_a \defeq \{\vec a:~\|\vec a - \vech a\|_2 \le \epsilon_2\}
$ 
for some $\epsilon_1, \epsilon_2 \ge 0
$ where $\|\cdot\|_F$ and $\|\cdot\|_2$ denote the matrix Frobenius norm and the vector $2$-norm, respectively, Problem \eqref{eq:robust-capon} can be equivalently transformed into \cite[Eq.~(29)]{vorobyov2003robust}, \cite[Eq.~(13)]{huang2023robust}
\begin{equation}\label{eq:robust-capon-explicit}
    \begin{array}{cl}
      \displaystyle \min_{\vec w} &  \vec w^\H (\math R + \epsilon_1 \mat I_N) \vec w \\
      \st  & \vec w^\H \vech a \ge \vec \epsilon_2 \|\vec w\|_2 + 1,
    \end{array}
\end{equation}
which is tantamount to the diagonal loading method in terms of the objective function. Problem \eqref{eq:robust-capon-explicit} is equivalent, in the sense of the same optimal cost, to
\begin{equation}\label{eq:robust-capon-explicit-2}
    \begin{array}{cl}
      \displaystyle \min_{\vec w} &  \vec w^\H (\math R + \epsilon_1 \mat I_N) \vec w \\
      \st  & |\vec w^\H \vech a| \ge \vec \epsilon_2 \|\vec w\|_2 + 1,
    \end{array}
\end{equation}
because if $\vec w^*$ is an optimal solution to \eqref{eq:robust-capon-explicit-2}, so is $\vec w^* e^{j\varphi}$ for any angle $\varphi$; $j$ denotes the imaginary unit; note that $\vec w^\H \vec w$ is rotation-invariant. For other proposals of $\cal U_a$, see, e.g., \cite[p.~871]{wu1999new}, \cite[p.~2408]{li2004doubly}, \cite[Eq.~(13)]{lorenz2005robust}, \cite[Eq.~(25)]{khabbazibasmenj2012robust}, \cite[Eq.~(10)]{gu2012robust}, \cite[p.~221]{huang2023robust}. Note that the robustification in terms of $\mat R$ and $\vec a$ can be independently conducted; cf. \eqref{eq:robust-capon}. Under the MVDR beamforming scheme, which occurs when the SoI is not included in the snapshots or when the IPN covariance $\Ripn$ is estimated using Methods 3a) or 3b) \cite{gu2012robust,mohammadzadeh2020maximum,ruan2013robust}, another interpretation of \eqref{eq:robust-capon} is to maximize the worst-case SINR \cite{kim2008robust,huang2023robust}; i.e., $\max_{\vec w} \min_{\vec a, \mat R} h(\vec w, \vec a, \mat R)$; cf. \eqref{eq:sinr}.

\textit{IPN Covariance Reconstruction}: The motivation behind this method is that the MVDR beamformer tends to be more robust than the MPDR beamformer \cite{ehrenberg2010sensitivity}. Hence, in the SoI-contaminated case, a natural way is to estimate the IPN covariance using snapshots \cite{gu2012robust,mohammadzadeh2020maximum,ruan2013robust,ruan2016robust}. In the typical treatment, the key step is to estimate the power spectra $P(\theta)$ and then reconstruct the IPN covariance as follows
\begin{equation}\label{eq:IPN}
    \Ripnh = \int_{\bar \Theta} P(\theta) \vec a_0(\theta) \vec a^\H_0(\theta) \d \theta,
\end{equation}
where $\bar \Theta$ is the angular sector that excludes the uncertainty region of the SoI's DoA \cite{gu2012robust,mohammadzadeh2020maximum}. In real-world operation, $P(\theta)$ is estimated using the Capon spectra \cite{gu2012robust} or the maximum-entropy spectra \cite{mohammadzadeh2020maximum}, and $\vec a_0(\theta)$ is estimated using $\vech a(\theta)$.

\subsection{Specific Research Questions}\label{subsec:research-questions}
From the review of existing works in Subsection \ref{subsec:existing-works}, we can see that the following questions have not been answered.
\begin{enumerate}[\text{Q}1)]
    \item According to \cite[p.~1365]{cox1987robust}, \cite[p.~313]{vorobyov2003robust}, \cite[p.~742]{rubsamen2011robust}, and our intuition, the well accepted notion of \quotemark{robust beamforming} means that the beamformer is insensitive to possible perturbations in signal characteristics. However, how does the worst-case optimization method in \eqref{eq:robust-capon} reflect this notion in terms of the uncertainty in $\math R$?

    \item Can we show that the existing eigenvalue thresholding method and the prior-knowledge embedding method are also tantamount to the worst-case optimization method \eqref{eq:robust-capon}, as the diagonal loading method is?

    \item In IPN covariance reconstruction, the resolution and bias of spectra estimation is crucial, especially for suppressing the closely located interferers. Therefore, how can we incorporate the higher-resolution and lower-bias spectra estimation methods, such as MUSIC, into IPN covariance reconstruction?
\end{enumerate}

The above three questions are representative particularizations of identified Issues I1)-I4) in Subsection \ref{subsec:problem-statements}. Note that Questions Q1) and Q2) apply to both MPDR and MVDR beamforming.

\section{Distributionally Robust Adaptive Beamforming}\label{sec:robustness-theory}
This section studies the formalized theory of distributionally robust (adaptive) beamforming, with a focus on combating the uncertainties in the estimated snapshot or IPN covariances. In particular, the concept of \quotemark{\textit{distributional robustness}} is quantitatively defined and the countermeasure methods are proposed. The concept of distributional robustness is noted because we are working with the probabilistic uncertainties in snapshots or IPN signals. This stochasticity feature fundamentally differs from the fixed (albeit unknown) uncertainty in the assumed steering vector $\vech a$.

In view of potential errors in the assumed steering vector $\vech a$, we consider the feasible set of beamformers $\vec w$ as
\begin{equation}\label{eq:W-def}
    \cal W \defeq \Big\{\vec w:~\displaystyle \min_{\vec a \in \cal U_a}  \vec w^\H \vec a \vec a^\H \vec w \ge 1 \Big\},
\end{equation}
for a given uncertainty set $\cal U_a$ of the steering vector. However, to focus on the uncertainty in $\math R$, the uncertainty in $\vech a$ is minimally examined in this article.

\subsection{Issue of Distributional Uncertainty}

We revisit the distributional form of beamforming
\begin{equation}\label{eq:capon-true}
    \begin{array}{cl}
      \displaystyle \min_{\vec w}  &  \vec w^\H \mat \E_{\rvec s \sim \Po}[\rvec s \rvec s^\H] \vec w \\
      \st  &  \vec w^\H \vec a_0 = 1,
    \end{array}
\end{equation}
where $\rvec s$ denotes the snapshot $\rvec x$ or the IPN signal under the MPDR and MVDR beamforming, respectively, and $\Po$ the underlying true distribution of $\rvec s$. Since $\Po$ is unavailable in practice, the empirical distribution 
\begin{equation}
    \Ph \defeq \frac{1}{L} \sum^L_{l = 1} \delta_{\vec s_l}
\end{equation}
constructed using $L$ collected samples $\{\vec s_1, \vec s_2, \ldots, \vec s_L\}$ can serve as an estimate of $\Po$ where $\delta_{\vec s}$ denotes the point-mass distribution centered at $\vec s$. As a result, the sample-average approximation (SAA) of \eqref{eq:capon-true} can be written as
\begin{equation}\label{eq:capon-saa}
    \begin{array}{cl}
      \displaystyle \min_{\vec w}  &  \vec w^\H \mat \E_{\rvec s \sim \Ph}[\rvec s \rvec s^\H] \vec w \\
      \st  &  \vec w^\H \vec a_0 = 1.
    \end{array}
\end{equation}
When $\cal U_a$ contains only $\vec a_0$,  
\begin{equation}\label{eq:capon-true-W}
      \displaystyle \min_{\vec w \in \cal W}  \vec w^\H \mat \E_{\rvec s \sim \Po}[\rvec s \rvec s^\H] \vec w
\end{equation}
reduces to \eqref{eq:capon-true} and
\begin{equation}\label{eq:capon-saa-W}
      \displaystyle \min_{\vec w \in \cal W}  \vec w^\H \mat \E_{\rvec s \sim \Ph}[\rvec s \rvec s^\H] \vec w
\end{equation}
reduces to \eqref{eq:capon-saa}. Since $\Ph$ is distributionally uncertain compared to $\Po$, directly employing \eqref{eq:capon-saa-W} as a surrogate of \eqref{eq:capon-true-W} is questionable due to the issue of \quotemark{overfitting on data}.

In beamforming, as only the second moments of $\Ph$ and $\Po$ are involved, for presentation simplicity, we directly work on $\math R \defeq \E_{\rvec s \sim \Ph}[\rvec s \rvec s^\H]$ and $\mat R_0 \defeq \E_{\rvec s \sim \Po}[\rvec s \rvec s^\H]$ whenever possible. Only when technically necessary, we investigate $\Ph$ and $\Po$. Another benefit of this treatment is that, under the MVDR beamforming scheme for SoI-contaminated cases, the IPN signal $\rvec s$ cannot be directly observed but its covariance can still be estimated \cite{gu2012robust,ruan2016robust}.

\subsection{Definition of Distributional Robustness}\label{subsec:robustness-def}
Motivated by the philosophical notion of robustness, we study the formal definitions of distributional robustness. We begin with the concept of local distributional robustness.

\begin{definition}[Locally Distributionally Robust Beamformer]\label{def:local-robust-beamformer}
A beamformer $\vec w^*$ is called $(\epsilon, k)$-locally-robust on the uncertainty set $\cal B_{\epsilon}(\mat R_0) \defeq \{\mat R:~d(\mat R, \mat R_0) \le \epsilon\}$ if 
\begin{equation}\label{eq:local-robustness}
    \vec w^{*\H} \mat R \vec w^*  - \vec w_0^\H  \mat R_0 \vec w_0 \le k,~~~~~\forall \mat R \in \cal B_\epsilon(\mat R_0),
\end{equation}
where
$
    \vec w_0 \in \min_{\vec w \in \cal W} \vec w^\H \mat R_0 \vec w
$ 
is an optimal beamformer associated with $\mat R_0$; $d$ is a matrix similarity measure. The smallest value $k^*$ of $k \ge 0$ satisfying \eqref{eq:local-robustness} is called the \textit{local robustness measure} of $\vec w^*$ at $\epsilon$.
\stp
\end{definition}

Definition \ref{def:local-robust-beamformer} means that when the real-world operating covariance $\mat R$ deviates from the underlying true $\mat R_0$, the performance degradation $\vec w^{*\H} \mat R \vec w^*  - \vec w_0^\H  \mat R_0 \vec w_0$ at the robust beamformer $\vec w^*$ is upper bounded by $k$ and $k^*$. This formalism straightforwardly reflects the notion of robustness, that is, the insensitivity in terms of perturbations in the covariance matrix. Note that the smaller the value of $k$ and $k^*$, the smaller the performance degradation under uncertainties, and therefore, the more robust the beamformer $\vec w^*$ is. 

Definition \ref{def:local-robust-beamformer} is natural if the deviation level $\epsilon$ is known, which, however, is not always the case in practice. Hence, the concept of global distributional robustness can be motivated.

\begin{definition}[Globally Distributionally Robust Beamformer]\label{def:global-robust-beamformer}
A beamformer $\vec w^*$ is called $(\tau, k)$-globally-robust on the whole space $\C^{N \times N}$ if 
\begin{equation}\label{eq:global-robustness}
    \vec w^{*\H} \mat R \vec w^*  - \tau \le k \cdot d(\mat R, \mat R_0),~~~~~\forall \mat R \in \C^{N \times N},
\end{equation}
where $\tau$ is a prescribed cost threshold satisfying
$
\vec w_0^\H  \mat R_0 \vec w_0 \le \tau < \infty. 
$ 
The smallest value $k^*$ of $k \ge 0$ satisfying \eqref{eq:global-robustness} is called the \textit{global robustness measure} of $\vec w^*$ at $\tau$.
\stp
\end{definition}

Definition \ref{def:global-robust-beamformer} means that when the real-world operating covariance $\mat R$ deviates from the underlying true $\mat R_0$, the performance degradation $\vec w^{*\H} \mat R \vec w^*  - \vec w_0^\H  \mat R_0 \vec w_0$ at the robust beamformer $\vec w^*$ is upper bounded by $\tau - \vec w_0^\H  \mat R_0 \vec w_0 + k \cdot d(\mat R, \mat R_0)$. For modeling flexibility, we do not restrict $\tau \equiv \vec w_0^\H  \mat R_0 \vec w_0$.

Hereafter in this article, for presentation brevity, we shorthand \quotemark{distributional robustness} simply as \quotemark{robustness}, if no ambiguity is caused. In addition, we do not explicitly mention whether the referred robustness is local or global; this can be straightforwardly inferred from context.

\subsection{Locally Distributionally Robust Beamforming}\label{subsec:local-robust-beamforming}
In line with Definition \ref{def:local-robust-beamformer} and its motivation, we propose the following robust beamforming formulation
\begin{equation}\label{eq:robust-bf-R0}
    \begin{array}{cl}
       (\vec w^*, k^*) = & \displaystyle \argmin_{\vec w \in \cal W, k} ~~ k \\
        \st & \vec w^\H \mat R \vec w  - \vec w_0^\H  \mat R_0 \vec w_0 \le k,~~~\forall \mat R \in \cal B_{\epsilon}(\mat R_0), \\
        & k \ge 0,
    \end{array}
\end{equation}
which finds the locally robust beamformer $\vec w^*$ and the local robustness measure $k^*$. 
In practice where $\mat R_0$ is inaccessible, we can solve \eqref{eq:robust-bf-R0} resorting to
\begin{equation}\label{eq:robust-bf-Rh}
    \begin{array}{cl}
       \displaystyle \min_{\vec w \in \cal W, k}  & k \\
        \st & \vec w^\H \mat R \vec w  - \vech w^\H \math R \vech w \le k,~~~\forall \mat R \in \cal B_{\delta}(\math R), \\
        & k \ge 0,
    \end{array}
\end{equation}
where $\math R$ is an estimate of $\mat R_0$ and 
\begin{equation}\label{eq:wh}
\vech w \in \min_{\vec w \in \cal W} \vec w^\H \math R \vec w
\end{equation}
is an optimal beamformer associated with $\math R$. The intuition is that if $\mat R_0 \in \cal B_{\delta}(\math R)$ for some $\delta \ge 0$, then there exists $\epsilon \ge 0$ such that $\math R \in \cal B_{\epsilon}(\mat R_0)$, and vice versa. As a result, $\math R$ is feasible to \eqref{eq:robust-bf-R0}, so is $\mat R_0$ to \eqref{eq:robust-bf-Rh}.

The theorem below reformulates Problems \eqref{eq:robust-bf-R0} and \eqref{eq:robust-bf-Rh}.
\begin{theorem}\label{thm:robust-bf-R0}
    Problem \eqref{eq:robust-bf-R0} is equivalent to 
    \begin{equation}\label{eq:minimax-robust-bf-R0}
           \displaystyle \min_{\vec w \in \cal W} \max_{\mat R \in \cal B_{\epsilon}(\mat R_0)}  \vec w^\H \mat R \vec w,
    \end{equation}
in the sense of the following: if $(\vec w^*, k^*, \mat R^*)$ solves \eqref{eq:robust-bf-R0}, then $(\vec w^*, \mat R^*)$ solves \eqref{eq:minimax-robust-bf-R0}; if $(\vec w^*, \mat R^*)$ solves \eqref{eq:minimax-robust-bf-R0}, by constructing $k^* \defeq \vec w^{*\H} \mat R^* \vec w^*  - \vec w_0^\H  \mat R_0 \vec w_0$, then $(\vec w^*, k^*, \mat R^*)$ solves \eqref{eq:robust-bf-R0}. In addition, Problem \eqref{eq:robust-bf-Rh} is equivalent to
\begin{equation}\label{eq:minimax-robust-bf-Rh}
       \displaystyle \min_{\vec w \in \cal W} \max_{\mat R \in \cal B_{\delta}(\math R)} \vec w^\H \mat R \vec w.
\end{equation}
\end{theorem}
\begin{proof}
See Appendix \ref{append:robust-bf-R0}.
\stp
\end{proof}

Theorem \ref{thm:robust-bf-R0} explains why minimax beamforming formulation is robust in the sense of Definition \ref{def:local-robust-beamformer}, which casts new insights into the robust beamforming community. Specifically, Theorem \ref{thm:robust-bf-R0} validates the rationale behind the worst-case optimization \eqref{eq:robust-capon}; see also \cite[Eq.~(9)]{kim2008robust}, \cite[Eq.~(46)]{shahbazpanahi2003robust}, and \cite[Eq.~(10)]{huang2023robust}. This answers Question Q1) in Subsection \ref{subsec:research-questions}.

Considering practicality, we particularly focus on the solution of \eqref{eq:minimax-robust-bf-Rh}.
\begin{theorem}\label{thm:sol-minimax-robust-bf-Rh}
If there exists $\mat R^* \in \cal B_{\delta}(\math R)$ such that $\mat R^* \succeq \mat R$ for all $\mat R \in \cal B_{\delta}(\math R)$, then 
\[
       \displaystyle \displaystyle \min_{\vec w \in \cal W} \max_{\mat R \in \cal B_{\delta}(\math R)} \vec w^\H \mat R \vec w = \min_{\vec w \in \cal W} \vec w^\H \mat R^* \vec w.
\]
\end{theorem}
\begin{proof}
This can be proven by contradiction.
\stp
\end{proof}

Specific applications of Theorem \ref{thm:sol-minimax-robust-bf-Rh} are given as follows, where eigenvalue thresholding, diagonal loading, regularization, and prior-knowledge embedding are shown to be robust; this answers Question Q2) in Subsection \ref{subsec:research-questions}.

\begin{example}[Eigenvalue Thresholding]\label{ex:eig-thres}
Consider 
\[
\cal B_{\mu}(\math R_{\text{thr}, \mu}) \defeq \{\mat R:~\mat 0 \preceq \mat R \preceq \math R_{\text{thr}, \mu}\},
\]
where $\math R_{\text{thr}, \mu}$ is defined in \eqref{eq:eig-thres}. There exists $\mat R^* \defeq \math R_{\text{thr}, \mu}$ such that $\mat R^* \succeq \mat R$ for all $\mat R \in \cal B_{\mu}(\math R_{\text{thr}, \mu})$. Hence, \eqref{eq:minimax-robust-bf-Rh} is particularized into
\begin{equation}\label{eq:et-obj}
\displaystyle \min_{\vec w \in \cal W} \vec w^\H \math R_{\text{thr}, \mu} \vec w,
\end{equation}
which is the eigenvalue-thresholding beamforming method \cite[Eq.~(12)]{lorenz2005robust}, \cite{harmanci2000relationships}.
\stp
\end{example}

\begin{example}[Diagonal Loading]\label{ex:diagonal-loading}
Consider 
\begin{equation}\label{eq:BDL-uncertainty-set}
\cal B_{\epsilon_1}(\math R) \defeq \{\mat R:~\math R - \epsilon_1 \mat I_N \preceq \mat R \preceq \math R + \epsilon_1\mat I_N,~\mat R \succeq \mat 0\},
\end{equation}
for $\epsilon_1 \ge 0$. There exists $\mat R^* \defeq \math R + \epsilon_1 \mat I_N$ such that $\mat R^* \succeq \mat R$ for all $\mat R \in \cal B_{\epsilon_1}(\math R)$. Hence, \eqref{eq:minimax-robust-bf-Rh} is particularized into
\begin{equation}\label{eq:dl-obj}
\displaystyle \min_{\vec w \in \cal W} \vec w^\H (\math R + \epsilon_1 \mat I_N) \vec w,
\end{equation}
which is the diagonal-loading beamforming method. In addition, we consider
\[
\cal U_a \defeq \{\vec a:~ \vech a \vech a^\H - \epsilon_2 \mat I_N \preceq \vec a \vec a^\H \preceq \vech a \vech a^\H + \epsilon_2 \mat I_N\},
\]
for $\epsilon_2 \ge 0$. Similarly, the constraint in $\cal W$ [see \eqref{eq:W-def}] can be explicitly expressed as 
\begin{equation}\label{eq:dl-constrj}
  \vec w^\H (\vech a \vech a^\H - \epsilon_2 \mat I_N) \vec w \ge 1,  
\end{equation}
that is, 
$
\cal W = \big\{\vec w:~\vec w^\H \vech a \vech a^\H \vec w \ge \epsilon_2 \vec w^\H\vec w + 1 \big\}
$. 
This is reminiscent of existing robust beamforming formulation \eqref{eq:robust-capon-explicit-2} in the literature; however, they are slightly distinct because square roots are involved in the constraint of \eqref{eq:robust-capon-explicit-2}.
\stp
\end{example}

Example \ref{ex:diagonal-loading} gives the diagonal loading method another robustness interpretation against the uncertainty in $\math R$ under the scheme of worst-case optimization, which is technically different from the results in \cite[Eq.~(46)]{shahbazpanahi2003robust}, \cite[Eq.~(10)]{huang2023robust}. This new interpretation brings new insights to the community. The motivation of constructing $\cal B_{\epsilon_1}(\math R)$ and $\cal U_a$ as in Example \ref{ex:diagonal-loading} is straightforward: we just assume that the matrix difference $\mat R - \math R$ is two-sided limited, so is the difference $\vec a \vec a^\H - \vech a \vech a^\H$. In \cite[Eq.~(51)]{shahbazpanahi2003robust}, \eqref{eq:dl-obj} is called positive diagonal loading, while \eqref{eq:dl-constrj} is negative diagonal loading. However, \cite{shahbazpanahi2003robust} obtains \eqref{eq:dl-obj} and \eqref{eq:dl-constrj} in a way technically different from our treatments as in Theorem \ref{thm:sol-minimax-robust-bf-Rh}. Example \ref{ex:diagonal-loading} can be generalized as follows.

\begin{example}[Regularization]\label{ex:regularization}
Consider 
\[
\cal B_{\epsilon_1}(\math R) \defeq \{\mat R:~\math R - \epsilon_1 \mat C_1 \preceq \mat R \preceq \math R + \epsilon_1 \mat C_1,~\mat R \succeq \mat 0\}
\]
and 
\[
\cal U_a \defeq \{\vec a:~ \vech a \vech a^\H - \epsilon_2 \mat C_2 \preceq \vec a \vec a^\H \preceq \vech a \vech a^\H + \epsilon_2 \mat C_2\},
\]
for Hermitian matrices $\mat C_1, \mat C_2 \succeq \mat 0$ and $\epsilon_1, \epsilon_2 \ge 0$. Problem \eqref{eq:minimax-robust-bf-Rh} is particularized into
\begin{equation}\label{eq:regularization-obj}
\displaystyle \min_{\vec w \in \cal W} \vec w^\H (\math R + \epsilon_1 \mat C_1) \vec w,
\end{equation}
where 
$
\cal W = \big\{\vec w:~\vec w^\H \vech a \vech a^\H \vec w \ge \epsilon_2 \vec w^\H \mat C_2 \vec w + 1 \big\}
$. 
This is a general regularized beamforming formulation.
\stp
\end{example}

\begin{example}[Prior-Knowledge Embedding]\label{ex:prior-knowledge}
Model \eqref{eq:regularization-obj} is equivalent, in the sense of the same optimal beamformer(s), to
\begin{equation}\label{eq:prior-knowledge-obj}
\displaystyle \min_{\vec w \in \cal W} \vec w^\H (\alpha \math R + \beta \mat C_1) \vec w,
\end{equation}
where $\alpha \defeq 1/(1 + \epsilon_1)$, $\beta \defeq \epsilon_1/(1 + \epsilon_1)$, and $\mat C_1$ can be seen as prior knowledge of the unknown $\mat R_0$. This gives the prior-knowledge embedding method \cite{stoica2008using}.
\stp
\end{example}

In addition to Examples \ref{ex:eig-thres}-\ref{ex:prior-knowledge}, another benefit of using Theorem \ref{thm:sol-minimax-robust-bf-Rh} can be seen in robust beamforming based on IPN covariance reconstruction; see Section \ref{sec:doa-estimation} later. In short, we can tailor the uncertainty set $\cal B_{\delta}(\math R)$ to achieve high-resolution power spectra estimation for better accuracy of IPN covariance reconstruction, especially when interferers are close to the SoI.

\subsection{Globally Distributionally Robust Beamforming}\label{subsec:global-robust-beamforming}
In line with Definition \ref{def:global-robust-beamformer} and its motivation, the globally robust beamforming problem can be proposed as
\begin{equation}\label{eq:robust-bf-R0-global}
    \begin{array}{cl}
       (\vec w^*, k^*) = &\displaystyle \argmin_{\vec w \in \cal W, k} ~~ k \\
        \st & \vec w^\H \mat R \vec w  - \tau \le k \cdot d(\mat R, \mat R_0),~\forall \mat R \in \C^{N \times N}, \\
        & k \ge 0,
    \end{array}
\end{equation}
which finds the globally robust beamformer $\vec w^*$ and the global robustness measure $k^*$. 
In real-world operation, \eqref{eq:robust-bf-R0-global} can be solved resorting to
\begin{equation}\label{eq:robust-bf-Rh-global}
    \begin{array}{cl}
       \displaystyle \min_{\vec w \in \cal W, k}  & k \\
        \st & \vec w^\H \mat R \vec w  - \tau \le k \cdot d(\mat R, \math R),~\forall \mat R \in \C^{N \times N}, \\
        & k \ge 0,
    \end{array}
\end{equation}
for 
\begin{equation}\label{eq:tau-h}
    \vech w^\H  \math R \vech w \le \tau < \infty.
\end{equation}
Note that $\math R$ is feasible to \eqref{eq:robust-bf-R0-global}, so is $\mat R_0$ to \eqref{eq:robust-bf-Rh-global}. Considering practicality, we particularly focus on the solution of \eqref{eq:robust-bf-Rh-global}. Let $\mat C$ be a Hermitian invertible weight matrix. We study three cases:
\begin{itemize}
    \item $d(\mat R, \math R) \defeq \Tr[\mat R - \math R]^\H[\mat R - \math R]$.
    
    \item $d(\mat R, \math R) \defeq \vec w^\H \Tr[\mat R - \math R]^\H[\mat R - \math R] \vec w$, $\forall \vec w \in \C^N$.

    \item $d(\mat R, \math R) \defeq \vec w^\H \Tr[\mat R - \math R]^\H \mat C^{-1} [\mat R - \math R] \vec w$, $\forall \vec w \in \C^N$.
\end{itemize}
In all cases, $d(\mat R, \math R)$ is a similarity measure between $\mat R$ and $\math R$: 1) $d(\mat R, \math R) \ge 0$ for every $\mat R$ and $\math R$; 2) $d(\mat R, \math R) = 0$ if and only if $\mat R = \math R$; 3) $d(\mat R, \math R) = d(\math R, \mat R)$. We choose $d(\mat R, \math R)$ in such ways just for technical tractability compared to, e.g., $d(\mat R, \math R) \defeq \|\mat R - \math R\|_F = \sqrt{\Tr[\mat R - \math R]^\H[\mat R - \math R]}$. Under the above three constructions for $d(\mat R, \math R)$, the solution to Problem \eqref{eq:robust-bf-Rh-global} is given below.

\begin{theorem}\label{thm:sol-robust-bf-Rh-global}
If $d(\mat R, \math R) \defeq \Tr[\mat R - \math R]^\H[\mat R - \math R]$, Problem \eqref{eq:robust-bf-Rh-global} is equivalent, in the sense of the same optimal beamformer(s), to a quartically regularized beamforming problem
\begin{equation}\label{eq:sol-global-robustness-1}
\min_{\vec w \in \cal W} \vec w^\H \left[\math R + \frac{\vec w \vec w^\H}{4k}\right] \vec w = \min_{\vec w \in \cal W} \vec w^\H \math R \vec w + \frac{1}{4k} (\vec w^\H \vec w)^2,
\end{equation}
where $k$ is chosen to let the above optimal objective value equal to $\tau$.
If $d(\mat R, \math R) \defeq \vec w^\H \Tr[\mat R - \math R]^\H[\mat R - \math R] \vec w$, Problem \eqref{eq:robust-bf-Rh-global} is equivalent, in the sense of the same optimal beamformer(s), to a quadratically regularized (i.e., diagonal-loading) beamforming problem
\begin{equation}\label{eq:sol-global-robustness-2}
\min_{\vec w \in \cal W} \vec w^\H \left[\math R + \frac{1}{4k} \mat I_N \right] \vec w = \min_{\vec w \in \cal W} \vec w^\H \math R \vec w + \frac{1}{4k} \vec w^\H \vec w,
\end{equation}
where $k$ is chosen to let the above optimal objective value equal to $\tau$.
If $d(\mat R, \math R) \defeq \vec w^\H \Tr[\mat R - \math R]^\H \mat C^{-1} [\mat R - \math R] \vec w$, Problem \eqref{eq:robust-bf-Rh-global} is equivalent, in the sense of the same optimal beamformer(s), to a general quadratically regularized beamforming problem
\begin{equation}\label{eq:sol-global-robustness-3}
\min_{\vec w \in \cal W} \vec w^\H \left[\math R + \frac{1}{4k} \mat C \right] \vec w = \min_{\vec w \in \cal W} \vec w^\H \math R \vec w + \frac{1}{4k} \vec w^\H \mat C \vec w,
\end{equation}
where $k$ is chosen to let the above optimal objective value equal to $\tau$.
\end{theorem}
\begin{proof}
See Appendix \ref{append:sol-robust-bf-Rh-global}.
\stp
\end{proof}

Since $\min_{\vec w} \vec w^\H [\math R + \frac{\vec w \vec w^\H}{4k}] \vec w$, $\min_{\vec w} \vec w^\H [\math R + \frac{1}{4k} \mat I_N] \vec w$, and $\min_{\vec w} \vec w^\H [\math R + \frac{1}{4k} \mat C] \vec w$ are continuous and monotonically decreasing in $k$ and they tend to infinity when $k \to 0$, the solutions to the three feasibility problems in Theorem \ref{thm:sol-robust-bf-Rh-global} exist. If we require $\tau > \vech w^\H  \math R \vech w$, the solutions are guaranteed to be finite (i.e., finite $k$'s exist). Note that when $\tau \defeq \vech w^\H  \math R \vech w$, we have $k = \infty$ for all the three cases. An intuitive example of Theorem \ref{thm:sol-robust-bf-Rh-global} is given as follows.
\begin{example}[Diagonal-Loading Beamforming]\label{ex:sol-global-robustness-2}
Let $\cal U_a \defeq \{\vec a_0\}$, that is, the steering vector is exactly known; cf. \eqref{eq:W-def}. Then the optimal objective of \eqref{eq:sol-global-robustness-2} is a function of $k$, i.e.,
\[
\varphi(k) \defeq \frac{1}{\vec a_0^\H \left[\math R + \frac{1}{4k} \mat I_N \right]^{-1} \vec a_0}.
\]
Let 
\[
    \tau \defeq \frac{1}{\vec a_0^\H \math R^{-1} \vec a_0} + t
\]
where $t \ge 0$ is a user-design objective-excess parameter; cf. \eqref{eq:tau-h}. Then, the optimal value of $k$ is given by the zero of the following equation
\[
\varphi(k) = \frac{1}{\vec a_0^\H \left[\math R + \frac{1}{4k} \mat I_N \right]^{-1} \vec a_0} = \frac{1}{\vec a_0^\H \math R^{-1} \vec a_0} + t = \tau;
\]
that is, $k$ is uniquely determined by $t$. If $t \defeq 0$, we have $k = \infty$.
\stp
\end{example}

As seen from Examples \ref{ex:diagonal-loading} and \ref{ex:sol-global-robustness-2} and Theorem \ref{thm:sol-robust-bf-Rh-global}, diagonal loading is a powerful technique for both locally robust beamforming and globally robust beamforming; this claim is true at least under decent algorithmic constructions. The difference is on which quantity the level of diagonal loading relies: the scale $\epsilon_1$ of the uncertainty set in locally robust beamforming (cf. Example \ref{ex:diagonal-loading}) or the prescribed threshold $\tau$ in globally robust beamforming (cf. Theorem \ref{thm:sol-robust-bf-Rh-global} and Example \ref{ex:sol-global-robustness-2}).

\subsection{Regularized Beamforming}\label{subsec:regularized-beamforming}

Existing literature has well-identified regularized beamforming as an efficient method to achieve robustness against limited snapshot sizes \cite{vorobyov2013principles,elbir2023twenty}, which can also be seen from this article's results in \eqref{eq:dl-obj}, \eqref{eq:regularization-obj}, \eqref{eq:sol-global-robustness-1}, \eqref{eq:sol-global-robustness-2}, and \eqref{eq:sol-global-robustness-3}. Another interpretation of regularized beamforming comes with penalizing the array sensitivity against array errors \cite{cox1987robust,rubsamen2011robust}. To be specific, the quantity $\vec w^\H \vec w$ is used as a sensitivity measure of the array; see \cite[Section~II]{cox1987robust} and \cite[Section~III-C]{rubsamen2011robust} for technical justifications. Intuitively speaking, when the steering vector $\vec a(\theta)$ deviates from its actual value $\vec a_0(\theta)$, the smaller the value of $\vec w^\H \vec w$, the less impact this deviation can cause to beamforming results; note that the derivative of array response pattern $\vec w^\H \vec a$ with respect to steering vector $\vec a$ is $\vec w$; that is, if the norm of $\vec w$ is small, the sensitivity of $\vec w^\H \vec a$ to deviations of $\vec a$ would be also limited. 

In this subsection, we introduce the regularization technique from a \textit{different} perspective, i.e., worst-case robustness. We start with the SAA beamforming \eqref{eq:capon-saa} with the noise injection technique \cite[Section~IV]{jablon1986adaptive}, \cite[Eq.~(13)]{pan2018noise} to achieve robustness. Suppose that the signal $\rvec s$ is contaminated by an unknown zero-mean random error $\rvec \xi_1$ (due to model uncertainties), and the operating array steering vector is contaminated by an unknown zero-mean random error $\rvec \xi_2$ (due to, e.g., calibration and pointing errors). The beamforming problem with noise injection into $\rvec s$ can be formulated as
\begin{equation}\label{eq:capon-stochastic}
    \begin{array}{cl}
      \displaystyle \min_{\vec w}  &  \vec w^\H \mat \E_{(\rvec s \sim \Ph,~\rvec \xi_1 \sim \P_1)}[(\rvec s + \rvec \xi_1) (\rvec s + \rvec \xi_1)^\H] \vec w \\
      
      \st  &  \vec w^\H \vech a \vech a^\H \vec w - \vec w^\H \E_{\rvec \xi_2 \sim \P_2}[ \rvec \xi_2 \rvec \xi_2^\H] \vec w \ge 1,
    \end{array}
\end{equation}
where we suppose that the distributions of $\rvec \xi_1$ and $\rvec \xi_2$ are $\P_1$ and $\P_2$, respectively; the objective is due to $\rvec \xi_1$-noise injection into $\rvec s$, while the constraint is due to the facts $\vech a = \vec a_0 + \rvec \xi_2$ and $\vec w^\H \vec a_0 \vec a_0^\H \vec w \ge 1$. Assume that $\rvec s$ is uncorrelated with $\rvec \xi_1$. Problem \eqref{eq:capon-stochastic} can be rewritten as
\begin{equation}\label{eq:capon-stochastic-R}
    \begin{array}{cl}
      \displaystyle \min_{\vec w}  &  \vec w^\H [\math R + \mat R_1] \vec w \\
      \st  &  \vec w^\H \vech a \vech a^\H \vec w - \vec w^\H \mat R_2 \vec w \ge 1,
    \end{array}
\end{equation}
where $\mat R_1 \defeq \E_{\rvec \xi_1 \sim \P_1} [\rvec \xi_1 \rvec \xi_1^\H]$ and $\mat R_2 \defeq \E_{\rvec \xi_2 \sim \P_2} [\rvec \xi_2 \rvec \xi_2^\H]$. In practice where we have no knowledge about $\mat R_1$ and $\mat R_2$, we can construct the uncertainty sets for $\mat R_1$ and $\mat R_2$ respectively as follows:
\begin{equation}\label{eq:uncerainty-set-R}
    \begin{array}{l}
        \mat R_1 \in \{\mat R:~\mat 0 \preceq \mat R \preceq \epsilon_1 \mat I_N\}, \\
        \mat R_2 \in \{\mat R:~\mat 0 \preceq \mat R \preceq \epsilon_2 \mat I_N\},
    \end{array}
\end{equation}
for some $\epsilon_1, \epsilon_2 \ge 0$. As a result, the robust counterpart of \eqref{eq:capon-stochastic-R}, which optimizes the worst-case performance and guarantees the worst-case feasibility, can be given as
\begin{equation}\label{eq:capon-robust-R}
    \begin{array}{cl}
      \displaystyle \min_{\vec w} \max_{\mat R_1}  &  \vec w^\H [\math R + \mat R_1] \vec w \\
      \st  &  \vec w^\H \vech a \vech a^\H \vec w - \displaystyle \max_{\mat R_2} \vec w^\H \mat R_2 \vec w \ge 1.
    \end{array}
\end{equation}
The solution to the above problem is given below.
\begin{corollary}[of Theorem \ref{thm:sol-minimax-robust-bf-Rh}]\label{cor:sol-capon-robust-R}
    Problem \eqref{eq:capon-robust-R} is equivalent to
    \begin{equation}\label{eq:sol-capon-robust-R}
        \begin{array}{cl}
          \displaystyle \min_{\vec w}  &  \vec w^\H [\math R + \epsilon_1 \mat I_N] \vec w \\
          \st  &  \vec w^\H \vech a \vech a^\H \vec w - \epsilon_2 \vec w^\H \vec w \ge 1,
        \end{array}
    \end{equation}
    which coincides with Example \ref{ex:diagonal-loading}.  \stp
\end{corollary}

Corollary \ref{cor:sol-capon-robust-R} and \eqref{eq:sol-capon-robust-R} reveal an important law in robust beamformer design: to achieve robustness is to add regularization. This law applies to both the objective and the constraint. To be specific, compared with Example \ref{ex:diagonal-loading} and Corollary \ref{cor:sol-capon-robust-R}, the insight below is immediate.
\begin{insight}[Robustness and Regularization]\label{insight:regularization}
Consider the SAA beamforming and its robust counterparts. The following two statements are at the core.
\begin{itemize}
    \item Solving $\min_{\vec w} \vec w^\H \math R \vec w$ by $\vech w$ at nominal value $\math R$ does not necessarily control true value $\vech w^\H \mat R_0 \vech w$ evaluated at $\mat R_0$. However, solving the regularized problem $\min_{\vec w} \vec w^\H [\math R + \epsilon_1 \mat I_N] \vec w$ by $\vech w'$ can control true value $\vech w^{\prime \H} \mat R_0 \vech w'$ evaluated at $\mat R_0$. This is because 
    \[
        \vec w^\H \mat R_0 \vec w \le \vec w^\H [\math R + \epsilon_1 \mat I_N] \vec w,~~~\forall \vec w.
    \]
    Hence, minimizing the regularized objective means minimizing the upper bound of the true objective function.
    \item Requiring $\vec w^\H \vech a \vech a^\H \vec w \ge 1$ does not necessarily imply $\vec w^\H \vec a_0 \vec a^\H_0 \vec w \ge 1$. However, if we require the regularized version, i.e., $\vec w^\H \vech a \vech a^\H \vec w - \epsilon_2 \vec w^\H \vec w \ge 1$, then we can guarantee that $\vec w^\H \vec a_0 \vec a^\H_0 \vec w \ge 1$ because
    \[
        \vec w^\H \vec a_0 \vec a^\H_0 \vec w \ge \vec w^\H \vech a \vech a^\H \vec w - \epsilon_2 \vec w^\H \vec w \ge 1,~~~\forall \vec w.
    \]
\end{itemize}
In short, for both objective and constraint, to achieve robustness is to add regularization.
\stp
\end{insight}

The following corollary solves Problem \eqref{eq:capon-robust-R} when the uncertainty sets for $\mat R_1$ and $\mat R_2$ take other forms than \eqref{eq:uncerainty-set-R}.
\begin{corollary}[of Corollary \ref{cor:sol-capon-robust-R}]\label{cor:sol-capon-robust-R-2}
If
\[
    \begin{array}{l}
        \mat R_1 \in \{\mat R:~\mat 0 \preceq \mat R \preceq \epsilon_1 \matb R_1\} \\
        \mat R_2 \in \{\mat R:~\mat 0 \preceq \mat R \preceq \epsilon_2 \matb R_2\}
    \end{array}
\] 
for some $\matb R_1, \matb R_2 \succeq \mat 0$, then Problem \eqref{eq:capon-robust-R} is equivalent to
    \begin{equation}\label{eq:sol-capon-robust-R-2}
        \begin{array}{cl}
          \displaystyle \min_{\vec w}  &  \vec w^\H [\math R + \epsilon_1 \matb R_1] \vec w \\
          \st  &  \vec w^\H \vech a \vech a^\H \vec w - \epsilon_2 \vec w^\H \matb R_2 \vec w \ge 1,
        \end{array}
    \end{equation}
which is a general regularized beamforming formulation than diagonal loading in Corollary \ref{cor:sol-capon-robust-R}; see also Example \ref{ex:regularization}.
 \stp
\end{corollary}

The benefit of using Corollary \ref{cor:sol-capon-robust-R-2} can be seen in robust beamforming based on IPN covariance reconstruction; see Section \ref{sec:doa-estimation} later. In short, we can design a good regularizer matrix $\matb R_1$ to achieve high-resolution power spectra estimation for better accuracy of IPN covariance reconstruction, especially when interferers are close to the SoI.

\subsection{Bayesian-Nonparametric Beamforming}\label{subsec:bayesian-beamforming}
In this subsection, we address the uncertainty in Problem \eqref{eq:capon-saa-W}, compared to Problem \eqref{eq:capon-true-W}, from the perspective of Bayesian nonparametrics. 

Let the uncertainty set for the distribution of $\rvec s$ be $\cal P$; for example, $\cal P$ can be constructed as a distributional ball $\cal B_\epsilon(\Ph)$ centered at $\Ph$. Since $\Ph$ is not a reliable surrogate of $\Po$, we assign a probability distribution $\Q$ on the measurable space $(\cal P, \cal B_{\cal P})$ where $\cal B_{\cal P}$ is the Borel $\sigma$-algebra on $\cal P$. As a result, the Bayesian-nonparametric counterpart of Problem \eqref{eq:capon-saa-W} can be given as
\begin{equation}\label{eq:capon-bayesian}
      \displaystyle \min_{\vec w \in \cal W}  \vec w^\H \E_{\P \sim \Q} \E_{\rvec s \sim \P}[\rvec s \rvec s^\H] \vec w,
\end{equation}
where $\Q$ is called a second-order probability distribution and $\P$ is a first-order one. Consider the measurable space $(\C^N, \cal B_{\C^N})$ where $\cal B_{\C^N}$ denotes the Borel $\sigma$-algebra on $\C^N$. For every event $\cal E \in \cal B_{\C^N}$, the random probability distribution $\P$ is a realization of $\Q$ and the quantity $\P(\cal E)$ is a random variable taking values on $[0, 1]$; the distribution of the random variable $\P(\cal E)$ is determined by $\Q$. A desired $\Q$ should let $\P$ concentrate around $\Po$; the more concentrated, the better. Formulation \eqref{eq:capon-bayesian} enables assigning a distribution on the covariance matrix and studying the following type of robust beamforming problem
$$
\min_{\vec w \in \cal W} \E_{\rmat R \sim \P_{\rmat R}} \vec w^\H \rmat R \vec w,
$$
where $\rmat R \defeq \E_{\rvec s \sim \P}[\rvec s \rvec s^\H]$ and $\P_{\rmat R}$ is its distribution; for instance, see \cite{huang2022robust,irani2025sinr}. 

In this article, we consider $\Q$ to be a Dirichlet process \cite{ferguson1973bayesian} with base distribution $\Pb \in \cal P$ and parameter $\alpha \ge 0$. Here, $\Pb$ serves as a prior estimate of $\Po$ (cf. the sample estimate $\Ph$). Using $\Q$ as a $(\Pb, \alpha)$-Dirichlet process means that for any finite $M$-partition $(\Xi_1, \Xi_2, \ldots, \Xi_M)$ of $\C^N$, the random vector $(\P(\Xi_1), \P(\Xi_2), \ldots, \P(\Xi_M))$ is distributed according to the Dirichlet distribution whose concentration parameter vector is $(\alpha\Pb(\Xi_1), \alpha\Pb(\Xi_2), \ldots, \alpha\Pb(\Xi_M))$. As a result of \cite[Chapter~3]{ghosal2017fundamentals}, Problem \eqref{eq:capon-bayesian} can be equivalently transformed to
\begin{equation}\label{eq:capon-bayesian-R}
      \displaystyle \min_{\vec w \in \cal W}  \vec w^\H \left[ \frac{L}{L + \alpha} \math R + \frac{\alpha}{L + \alpha} \matb R \right] \vec w,
\end{equation}
where $\matb R \defeq \E_{\rvec s \sim \Pb} [\rvec s \rvec s^\H]$ is the prior knowledge. Intuitively speaking, in \eqref{eq:capon-bayesian-R}, we expect the combined probability distribution $[\frac{L}{L + \alpha} \Ph + \frac{\alpha}{L + \alpha} \Pb]$ to be a better estimate of $\Po$ than $\Ph$ and $\Pb$, so is $[ \frac{L}{L + \alpha} \math R + \frac{\alpha}{L + \alpha} \matb R]$ to $\mat R_0$ than $\math R$ and $\matb R$. 

Model \eqref{eq:capon-bayesian-R} gives the prior-knowledge embedding method \cite{stoica2008using} another rigorous interpretation (cf. Example \ref{ex:prior-knowledge}), which offers new insights to the community. In addition, it is equivalent, in the sense of the same optimal beamformer(s), to
\begin{equation}\label{eq:capon-bayesian-regularized}
      \displaystyle \min_{\vec w \in \cal W}  \vec w^\H [\math R + \epsilon \matb R ] \vec w = \displaystyle \min_{\vec w \in \cal W}  \vec w^\H \math R \vec w + \epsilon \vec w^\H \matb R \vec w,
\end{equation}
where $\epsilon \defeq {\alpha}/{L}$. Model \eqref{eq:capon-bayesian-regularized} is a regularized beamforming problem; cf. Example \ref{ex:regularization} and Corollary \ref{cor:sol-capon-robust-R-2}.

A direct application of Models \eqref{eq:capon-bayesian-R} and \eqref{eq:capon-bayesian-regularized} is as follows.
\begin{example}[Diagonal Loading]\label{ex:diag-loading-bayes}
    When we do not know the prior knowledge $\matb R$, we can set $\matb R \defeq \mat I_N$. Problem \eqref{eq:capon-bayesian-R} becomes 
    \[
        \displaystyle \min_{\vec w \in \cal W}  \vec w^\H \left[ \frac{L}{L + \alpha} \math R + \frac{\alpha}{L + \alpha} \mat I_N \right] \vec w,
    \]
    and Problem \eqref{eq:capon-bayesian-regularized} turns into 
    \[
        \displaystyle \min_{\vec w \in \cal W}  \vec w^\H [\math R + \epsilon \mat I_N ] \vec w.
    \]
    The latter is the diagonal-loading method. \stp
\end{example}

Since $\alpha \ge 0$ is a user-design parameter in real-world operation, we can write \eqref{eq:capon-bayesian-R} more compactly as
\begin{equation}\label{eq:capon-bayesian-compact}
      \displaystyle \min_{\vec w \in \cal W}  \vec w^\H [(1-\beta) \math R + \beta \matb R ] \vec w,
\end{equation}
through employing another user-design parameter $\beta \in [0, 1]$. The following insight presents the robustness of the Bayesian-nonparametric beamforming model \eqref{eq:capon-bayesian-compact}.
\begin{insight}[Bayesian-Nonparametric Beamforming]\label{insight:bayes-model}
Suppose that under $(\matb R, \beta)$, the convex combination $(1-\beta) \math R + \beta \matb R$ is closer to $\mat R_0$ than $\math R$. Then, for $\epsilon_1, \epsilon_2 \ge 0$ such that $(1-\beta) \math R + \beta \matb R \in \cal B_{\epsilon_1}({\mat R_0})$ and $\math R \in \cal B_{\epsilon_2}({\mat R_0})$, we have $\epsilon_1 \le \epsilon_2$. As a result, if $k_1$ solves Problem \eqref{eq:robust-bf-R0} under $\epsilon \defeq \epsilon_1$ and $k_2$ solves Problem \eqref{eq:robust-bf-R0} under $\epsilon \defeq \epsilon_2$, we have $k_1 \le k_2$. Therefore, compared with the case using $\math R$, employing $(1-\beta) \math R + \beta \matb R$ as a surrogate of $\mat R_0$ would generate a more robust beamformer that has a smaller robustness measure; see Definition \ref{def:local-robust-beamformer}. \stp
\end{insight}

In addition to Insight \ref{insight:bayes-model}, another benefit of using Bayesian-nonparametric beamforming method \eqref{eq:capon-bayesian-compact} can be seen in robust beamforming based on IPN covariance reconstruction; see Section \ref{sec:doa-estimation} later. In short, we can design the prior matrix $\matb R$ to achieve high-resolution power spectra estimation for better accuracy of IPN covariance reconstruction, especially when interferers are close to the SoI.

\subsection{Unified Distributionally Robust Beamforming Framework}\label{subsec:unified-framework}
Based on the proposals of locally distributionally robust beamforming \eqref{eq:robust-bf-Rh}, globally distributionally robust beamforming \eqref{eq:robust-bf-Rh-global}, regularized beamforming \eqref{eq:sol-capon-robust-R-2}, and Bayesian-nonparametric beamforming \eqref{eq:capon-bayesian-compact}, the insight below summarizes a unified distributionally robust beamforming framework.
\begin{insight}[Unified Robust Beamforming Framework]\label{insight:unified-framework}
A robust beamforming formulation in the sense of Definition \ref{def:local-robust-beamformer} should take, or can be transformed into, the following forms: 
\[
      \displaystyle \min_{\vec w \in \cal W}  \vec w^\H [(1-\beta) \math R + \beta \matb R ] \vec w
\]
or equivalently
\[
      \displaystyle \min_{\vec w \in \cal W}  \vec w^\H \math R \vec w + \epsilon \vec w^\H \matb R \vec w,
\]
where $\beta \in [0, 1]$ and $\epsilon \defeq \beta/(1-\beta) \ge 0$; the parameters $\beta$, $\epsilon$, and $\matb R$ can be determined according to the principles of locally robust beamforming (i.e., Theorems \ref{thm:robust-bf-R0} and \ref{thm:sol-minimax-robust-bf-Rh}), globally robust beamforming (i.e., Theorem \ref{thm:sol-robust-bf-Rh-global}), regularized beamforming (i.e., Insight \ref{insight:regularization}), or Bayesian-nonparametric beamforming (i.e., Insight \ref{insight:bayes-model}).
\stp
\end{insight}

For detailed motivations and explanations, also recall Examples \ref{ex:diagonal-loading}, \ref{ex:regularization}, \ref{ex:prior-knowledge}, \ref{ex:sol-global-robustness-2}, and \ref{ex:diag-loading-bayes}, Corollaries \ref{cor:sol-capon-robust-R} and \ref{cor:sol-capon-robust-R-2}, and Model \eqref{eq:capon-bayesian-compact}.

\section{Robustness via High-Resolution Spectra Estimation and IPN Covariance Reconstruction}\label{sec:doa-estimation}

This section answers Question Q3) in Subsection \ref{subsec:research-questions}. Since we aim to estimate the power spectra for IPN covariance reconstruction, in this section, $\math R$ refers to the snapshot covariance in which the SoI is present. This leads to the Capon (i.e., MPDR) beamforming and Capon power spectra.

The Capon beamforming is believed to be a good-resolution method for power spectra estimation \cite{johnson1982application}. However, when we move to the robust Capon beamforming regime, the resolution must be sacrificed. This can be straightforwardly seen from the robust counterpart of the Capon constraint \eqref{eq:W-def}, i.e.,
$
    \min_{\vec a(\theta) \in \cal U_a}  \vec w^\H \vec a(\theta) \vec a(\theta)^\H \vec w \ge 1,
$ 
for 
$
    \cal U_a \defeq \{\vec a:~\|\vec a - \vec a_0(\theta)\| \le \epsilon\}
$. 
To be specific, the robust Capon beamforming maintains a wider beam than the standard Capon beamforming in direction $\theta$ to ensure that the array power response $\vec w^\H \vec a_0(\theta') \vec a_0(\theta')^\H \vec w$ is sufficiently large for all directions $\theta'$ close to $\theta$. Note that
$\vec a_0(\theta') \in \cal U_a$ when $|\theta' - \theta|$ is small. In using beamforming for power spectra estimation, the wider the beam, the lower the resolution. Therefore, in the literature, the trade-off between robustness and resolution in beamforming for power spectra estimation largely exists.  
This dilemma occurs even when the steering vector is exactly known but the SAA-estimate $\math R$ deviates from its true value $\mat R_0$. This is because robustness against array errors technically leads to the diagonal-loading Capon beamforming \cite{cox1987robust,vorobyov2003robust,li2003robust}, so does robustness against the uncertainty in $\math R$ (see, e.g., Example \ref{ex:diagonal-loading}). Therefore, whenever the diagonal loading technique is applied to achieve robustness, the resolution for power spectra estimation must deteriorate. In this section, we show how to leverage the characteristics of the subspace methods, such as MUSIC, to improve the resolution of power spectra estimation in robust beamforming. This aim is technically realized by designing an appropriate $\matb R$ in Insight \ref{insight:unified-framework}.

We start with the reformulation of the canonical MUSIC method. Let
\begin{equation}\label{eq:eigen}
   \mat U_0 \mat \Sigma_0 \mat U_0^\H \defeq \mat R_0
\end{equation}
denote an eigenvalue decomposition of $\mat R_0$ and eigenvalues in $\mat \Sigma_0$ be sorted in the descending order. We partition $\mat U_0$ and $\mat \Sigma_0$ as follows:
\[
\mat \Sigma_0 = 
\left[
\begin{array}{cc}
   \mat \Sigma_{0,s}  &  \mat 0\\
   \mat 0          & \mat \Sigma_{0,v}
\end{array}
\right]
\]
where $\mat \Sigma_{0,s} \in \C^{K \times K}$ and $\mat \Sigma_{0,v} \in \C^{(N-K) \times (N-K)}$, and
\[
\mat U_0 = [\mat U_{0,s},~\mat U_{0,v}],
\]
where the columns of $\mat U_{0,s} \in \C^{N \times K}$ span the signal-plus-interference subspace and the columns of $\mat U_{0,v} \in \C^{N \times (N-K)}$ span the noise subspace. Let
\begin{equation}\label{eq:eigen-Rh}
   \math U \math \Sigma \math U^\H \defeq \math R
\end{equation}
denote an eigenvalue decomposition of $\math R$ and eigenvalues in $\math \Sigma$ be sorted in the descending order. Similarly, $\math \Sigma$ is partitioned into $\math \Sigma_s$ and $\math \Sigma_v$, so is $\math U$ into $\math U_{s}$ and $\math U_{v}$.

The MUSIC pseudo-spectra are defined as $P_{\text{M}}(\theta) \defeq 1/[\vec a^\H_0(\theta) \mat U_{0,v} \mat U^\H_{0,v} \vec a_0(\theta)]$, that is,
\[
     P_{\text{M}}(\theta) = 
    \frac{\sigma^{-2}_n}{\vec a^\H_0(\theta) 
    \left[
        \mat U_{0,s},~\mat U_{0,v}
    \right] 
    \left[
        \begin{array}{cc}
             \mat 0 &  \\
                    & \mat \Sigma^{-1}_{0,v}
        \end{array}
    \right]
    \left[
        \begin{array}{c}
             \mat U^\H_{0,s}  \\
             \mat U^\H_{0,v}
        \end{array}
    \right]
    \vec a_0(\theta)
    }.
\]
Note that $\mat \Sigma_{0,v} = \sigma^2_n \mat I_{N-K}$ where $\sigma^2_n$ denotes the power of channel noise. Because the constant $\sigma^{-2}_n$ does not influence the power pattern, we can ignore it and modify the MUSIC pseudo-spectra to
\begin{equation}\label{eq:spectra-music}
     P_{\text{M}}(\theta) = 
    \frac{1}{\vec a^\H_0(\theta) 
    \left[
        \mat U_{0,s},~\mat U_{0,v}
    \right] 
    \left[
        \begin{array}{cc}
             \mat 0 &  \\
                    & \mat \Sigma^{-1}_{0,v}
        \end{array}
    \right]
    \left[
        \begin{array}{c}
             \mat U^\H_{0,s}  \\
             \mat U^\H_{0,v}
        \end{array}
    \right]
    \vec a_0(\theta)
    }.
\end{equation}

In contrast, the Capon power spectra, which solve Problem \eqref{eq:capon}, are given by $P(\theta) = 1/[\vec a^\H_0(\theta) \mat R_0^{-1} \vec a_0(\theta)]$, that is,
\begin{equation}\label{eq:capon-spectra}
    P(\theta) = \frac{1}{\vec a^\H_0(\theta) 
    \left[
        \mat U_{0,s},~\mat U_{0,v}
    \right] 
    \left[
        \begin{array}{cc}
             \mat \Sigma^{-1}_{0,s} & \\
              & \mat \Sigma^{-1}_{0,v}
        \end{array}
    \right]
    \left[
        \begin{array}{c}
             \mat U^\H_{0,s}  \\
             \mat U^\H_{0,v}
        \end{array}
    \right]
    \vec a_0(\theta)
    }.
\end{equation}

Comparing MUSIC pseudo-spectra in \eqref{eq:spectra-music} with Capon spectra in \eqref{eq:capon-spectra}, we have the following key observation.
\begin{insight}\label{insight:high-resolution}
In Capon beamforming and its variants, from the perspective of algorithm design, the key to high-resolution power spectra estimation is to downweight the signal-plus-interference subspace component $\vec a^\H_0(\theta) \mat U_{0,s} \mat \Sigma^{-1}_{0,s} \mat U^\H_{0,s} \vec a_0(\theta)
$. One way is to modify $\mat \Sigma^{-1}_{0,s}$ in \eqref{eq:capon-spectra} to $\gamma \mat \Sigma^{-1}_{0,s}$ where $0 \le \gamma \le 1$ is a small value. Another choice is to modify $\mat \Sigma^{-1}_{0,s}$ in \eqref{eq:capon-spectra} to $[\mat \Sigma_{0,s} + \gamma \mat I_K]^{-1}$ where $\gamma \ge 0$ is a large value.
\stp
\end{insight}

\subsection{Robust Beamforming for High-Resolution Power Spectra Estimation}\label{subsec:beamforming-high-resolution}

Combining the principles of robustness in Insight \ref{insight:unified-framework} and high resolution in Insight \ref{insight:high-resolution}, the robust beamforming methods for high-resolution power spectra estimation can be put forward. Let
\begin{equation}\label{eq:Gamma}
    \mat \Gamma \defeq
    \math U \left[
    \begin{array}{cc}
       \delta_1 \mat I_{K'}  &  \\
         & \delta_2 \mat I_{N-K'}
    \end{array}
    \right] \math U^\H,
\end{equation}
for $\delta_1, \delta_2 \ge 0$ and $K \le K' \le N - 1$. 

Consider the robust beamforming problem
\begin{equation}\label{eq:method-1}
       \displaystyle \min_{\vec w \in \cal W} \vec w^\H (\math R + \mat \Gamma) \vec w.
\end{equation}
The above formulation can be seen as an \bfit{unbalanced diagonal loading} (UDL) method, which generalizes the conventional balanced diagonal loading (BDL) method where $\delta_1 = \delta_2$. Note that for the purpose of high-resolution power spectra estimation, $K'$ can be any integer value in $[K, N - 1]$; see \cite[Section~V]{johnson1982application}, \cite[Section~4.5]{stoica2005spectral} for explanations and justifications. Note also that for the purpose of robustness, we only expect $\mat \Gamma$ to be positive semi-definite such that $\vec w^\H \mat R_0 \vec w \le \vec w^\H \math R \vec w + \vec w^\H \mat \Gamma \vec w$, for every $\vec w$, regardless of the value of $K'$; cf. Insights \ref{insight:regularization} and \ref{insight:unified-framework}.

Suppose $\cal U_a \defeq \{\vec a_0(\theta)\}$; that is, the steering vector is exactly known; cf. \eqref{eq:W-def}. We have the following observations:
\begin{itemize}
    \item If $\delta_1 = \delta_2 = 0$, Model \eqref{eq:method-1} gives the standard Capon beamforming method.

    \item If $\delta_1 = \delta_2 \neq 0$, Model \eqref{eq:method-1} gives the conventional BDL method; i.e., the two loading levels are the same.

    \item If $\delta_1 \neq \delta_2$, Model \eqref{eq:method-1} presents a UDL method; i.e., the two loading levels are different.

    \item For any $K' \in [K, N-1]$, if $\delta_1 \to \infty$, Model \eqref{eq:method-1} generates the high-resolution subspace pseudo-spectra \cite[Section~5]{johnson1982application}. When $K' = K$, this subspace pseudo-spectra become the MUSIC pseudo-spectra \eqref{eq:spectra-music}. When $K' = N-1$, this subspace pseudo-spectra become the Pisarenko pseudo-spectra \cite[p.~162]{stoica2005spectral}.
\end{itemize}
Motivated by the preceding observations, we propose the following robust beamforming for high-resolution power spectra estimation.
\begin{method}
\label{method:robust-capon-beamforming}
Solve robust Capon beamforming \eqref{eq:method-1} with $\mat \Gamma$ in \eqref{eq:Gamma}. To achieve high resolution for power spectral estimation, let $\delta_1 \gg \delta_2 \ge 0$; to achieve robustness against the uncertainty in $\math R$, let $\delta_1,\delta_2 > 0$.
\stp
\end{method}

\subsection{Robust Beamforming via IPN Covariance Reconstruction}
Motivated by Insight \ref{insight:unified-framework} and Method \ref{method:robust-capon-beamforming}, this subsection summarizes the algorithm for robust beamforming based on high-resolution power spectra estimation and IPN covariance reconstruction; see Algorithm \ref{algo:bf-doa}. In practice, since $\delta_1$ and $\delta_2$ in \eqref{eq:Gamma} are still user-design parameters, the employment of tradeoff parameters $\epsilon_1$ in \eqref{eq:sol-capon-robust-R-2} and $\beta$ in \eqref{eq:capon-bayesian-compact} are dispensable. Therefore, in practice, we can set $\epsilon_1 \defeq 1$ and $\beta \defeq 0.5$. 

\begin{algorithm}[!htbp]
    \caption{Robust beamforming based on high-resolution power spectra estimation and IPN covariance reconstruction}
    \label{algo:bf-doa}
    \begin{flushleft}
        \justifying
        \textbf{Basics}: Insight \ref{insight:unified-framework}, Insight \ref{insight:high-resolution}, and Method \ref{method:robust-capon-beamforming}.\\
        \textbf{Notes}: In defining $\mat \Gamma$, to ensure high resolution and sufficient robustness, let $\delta_1 \gg \delta_2 > 0$. In evaluating \eqref{eq:IPN}, the integration is approximated using a summation; cf. \cite[Eq.~(20)]{mohammadzadeh2020maximum}.
    \end{flushleft}
    \begin{algorithmic}
        \Require $\math R$, $\mat \Gamma$, $\vech a$, $\epsilon$
        \State // \text{Step 1: High-Resolution Power Spectral Estimation}
        \State \quad $P(\theta) = \frac{1}{\vech a^{\H}(\theta) [\math R + \mat \Gamma]^{-1} \vech a(\theta)}$ 

        \State // \text{Step 2: IPN Covariance Reconstruction}
        \State \quad Use $P(\theta)$, $\vech a$, and \eqref{eq:IPN} to reconstruct $\Ripnh$

        \State // \text{Step 3: Robust Beamforming using Diagonal Loading}
        \State \quad $\vec w^* = (\Ripnh + \epsilon \mat I_N)^{-1} \vech a / [\vech a^\H (\Ripnh + \epsilon \mat I_N)^{-1} \vech a]$
        
        \Ensure {Robust beamformer \captext{$\vec w^*$}}
    \end{algorithmic}
\end{algorithm}

In Step 1, the unbalanced diagonal loading operation does not alter the main characteristics of the power spectrum (e.g., the locations of peaks), since it perturbs only the eigenvalues while preserving the eigenvectors; cf. \eqref{eq:Gamma}. Although eigenvalues affect the shape and amplitude scaling of the spectrum (thereby allowing control over resolution), the angular selectivity in a power spectrum is primarily governed by the eigenvectors \cite{schmidt1986multiple}; for empirical illustration, see Fig.~\ref{fig:capon-spectra} in the experiment section. Therefore, the directions of interferences and their relative power levels, as reflected in the spectrum, can be effectively identified. As a result, the IPN covariance reconstruction in Step 2 remains valid and captures the essential statistical information of the IPN. Consequently, in Step 3, the IPN signals can be effectively suppressed. 

The computational complexity of Algorithm \ref{algo:bf-doa} is as follows.

\begin{remark}[Computationally Complexity]\label{rem:complexity}
The computational complexity of Algorithm \ref{algo:bf-doa} is comparable to the method in \cite{gu2012robust}. The differences between the two approaches are two-fold: first, when estimating the power spectra, Algorithm \ref{algo:bf-doa} uses $\math R + \mat \Gamma$, instead of $\math R$, to enhance resolution and robustness; second, diagonal loading, i.e., $\Ripnh + \epsilon \mat I_N$, is applied to ensure the robustness against the uncertainty in $\Ripnh$. The two differences lead to additional $\cal O(N^2)$ FLOPs (floating-point operations). Hence, the main computational burden of Algorithm \ref{algo:bf-doa} is at evaluating the integral in \eqref{eq:IPN}, which admits $\cal O(N^2 S)$ FLOPs; $S$ is the number of sampling points in $\bar \Theta$; NB: $S \gg N$. As a result, the overall complexity of Algorithm \ref{algo:bf-doa} is $\cal O(N^2 S)$, the same as the method in \cite{gu2012robust}.
\stp
\end{remark}

In addition, the following remark is practically instructive.
\begin{remark}[Design $\cal W$]\label{rem:W}
Technically, addressing either the uncertainty in the assumed steering vector or that in the estimated IPN covariance leads to the regularization operation (e.g., diagonal loading method); see \cite{cox1987robust,vorobyov2003robust,li2003robust}, Insight \ref{insight:regularization}, and Insight \ref{insight:unified-framework}. Therefore, in real-world operation, there is no need to consider the uncertainty in the steering vector; that is, the uncertainty set for the steering vector contains only a singleton:
\begin{equation}\label{eq:practical-W}
    \cal U_a(\theta) \defeq \{\vech a(\theta)\},~~~~~\forall \theta;
\end{equation}
see \eqref{eq:W-def}. This trick greatly reduces the operational complexity of Algorithm \ref{algo:bf-doa}. However, the robustness against the uncertainties in the assumed steering
vector and the estimated IPN covariance remains. If general $\cal U_a$ are preferred, refer to handling methods in, e.g.,  \cite{wu1999new,vorobyov2003robust,li2004doubly,lorenz2005robust,huang2023robust}.
\stp
\end{remark}

\section{Experiments}\label{sec:experiment}
In this section, we consider a uniform linear array with $N = 10$ and half-wavelength spacing \cite{vorobyov2003robust,gu2012robust}. We show how the unbalanced diagonal loading trick in Method \ref{method:robust-capon-beamforming} benefits the suppression of closely located interferers. Our method complements the state of the art because existing studies assume that interferers are sufficiently separated in their experiments; see, e.g., \cite{vorobyov2003robust,gu2012robust,mohammadzadeh2022covariance}. We suppose that the SoI is at $\theta_1 = -30^\circ$, and two \textbf{strong} interferers with interference-to-signal ratio (ISR) of 10dB are located at $-22^\circ$ and $30^\circ$. In IPN covariance reconstruction, we let the integration interval in \eqref{eq:IPN} be $\bar \Theta \defeq [-90^\circ, -35^\circ] \cup [-25^\circ, 90^\circ]$ (cf. \cite{gu2012robust}); that is, the interferer at $\theta_2 = -22^\circ$ is close to the boundary of the uncertainty region $[-35^\circ, -25^\circ]$ of the SoI. Throughout the experiments, we let $\vech a(\theta) \defeq \vec a_0(\theta) + \vec \Delta$ where $\vec \Delta$ is a Gaussian random vector with covariance $0.01 \mat I_N$; for every Monte-Carlo episode, this $\vec \Delta$ changes its value. The output SINR performance of each beamformer is averaged over 500 Monte-Carlo simulations. All the source data and codes are available online at GitHub: \url{https://github.com/Spratm-Asleaf/Beamforming-UDL}.


\subsection{Experimental Verification of Robustness Principles}
In this subsection, we provide empirical verifications of the two robustness principles (i.e., Bayesian and Regularization) in Insight \ref{insight:unified-framework}. Without loss of generality for demonstration, the MPDR (i.e., Capon) beamforming scheme is employed. Accordingly, we refer to the first robustness formulation as Capon-Bayesian, while the second as Capon-Regularization. We sweep the value of $\beta$ from 0 to 1 with step size 0.2, and fix $\epsilon = 0.25$ so that, for $\beta = 0.2$, we have $\epsilon = 0.25 = 0.2/0.8 = \beta/(1-\beta)$. The output SINR is plotted against input signal-to-noise ratio (SNR). The number $L$ of snapshots is 30. (Other values of $\epsilon$ and $L$ do not change the empirical claims in this subsection.) For performance comparison, the following two methods are also implemented: 
\begin{itemize}
    \item {Optimal}: The optimal beamformer where the true IPN covariance matrix $\Ripn$ and the true steering vector $\vec a_0(\theta)$ are exactly known. Note that, in real-world operations, this method is not applicable.
    
    \item {Capon}: The Capon beamformer with estimated snapshot covariance $\math R$ and assumed steering vector $\vech a(\theta)$.
\end{itemize}
First, we assume that no effective prior information about the true snapshot covariance $\mat R_0$ is available, so we let $\matb R \defeq \mat I_N$. Consequently, the second principle in Insight \ref{insight:unified-framework}, i.e., the regularization formulation $\min_{\vec w \in \cal W}  \vec w^\H \math R \vec w + \epsilon \vec w^\H \matb R \vec w$, specifies the diagonal loading method with loading level $\epsilon = 0.25$, which provides the robustness against steering-vector uncertainties and background white noises. The experimental results are shown in Fig. \ref{fig:robustness-sinr}, for which we have the following remarks:
\begin{itemize}
    \item In the Bayesian non-parametric principle, convexly combining the empirical estimate $\math R$ with the prior guess $\matb R$ can potentially improve the robustness.

    \item In the regularization principle, employing the regularizer $\matb R$ can also potentially improve the robustness.

    \item When $\beta = 0.0$, Capon-Bayesian reduces to the Capon method; when $\beta = 0.2$, Capon-Bayesian is equivalent to Capon-Regularization because $\epsilon = 0.25 = \beta/(1-\beta)$; when $\beta = 1.0$, Capon-Bayesian reduces to the conventional non-adaptive Bartlett beamformer.

    \item When $\beta$ is well tuned (e.g., $0.4, 0.6, 0.8$), Capon-Bayesian outperforms Capon-Regularization with $\epsilon = 0.25$.
\end{itemize}

\begin{figure}[!htbp]
	\centering
	\subfigure[$\beta = 0.0$]{
	 	\includegraphics[height=3cm]{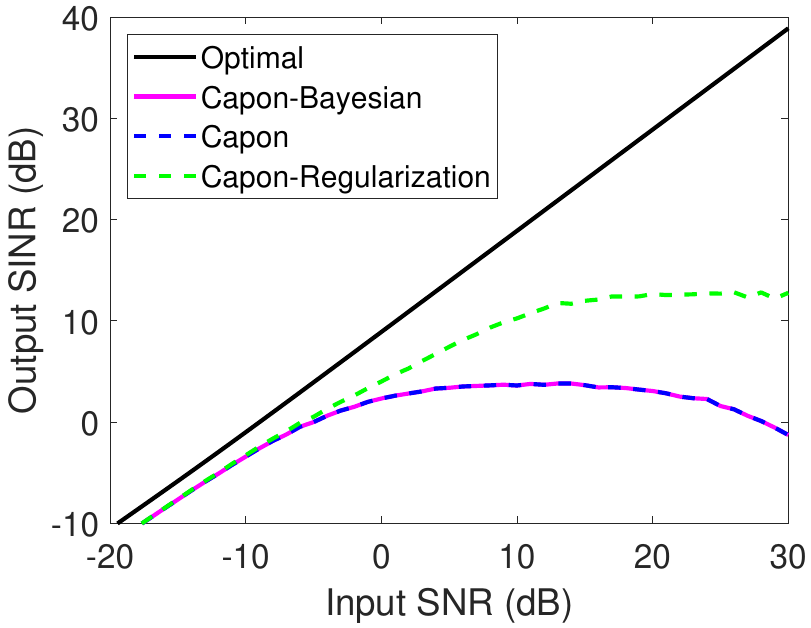}
	}
    \subfigure[$\beta = 0.2$]{
	 	\includegraphics[height=3cm]{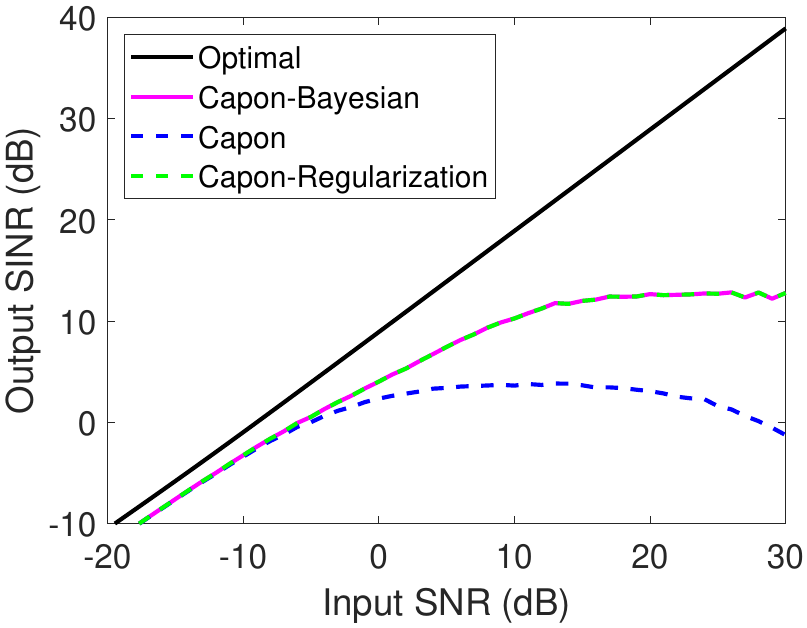}
	}

	\subfigure[$\beta = 0.4$]{
	 	\includegraphics[height=3cm]{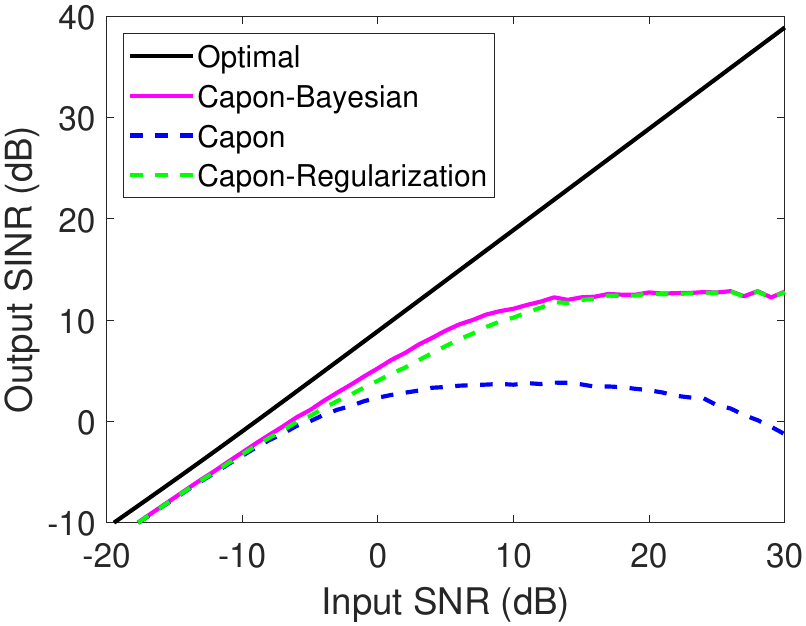}
	}
    \subfigure[$\beta = 0.6$]{
	 	\includegraphics[height=3cm]{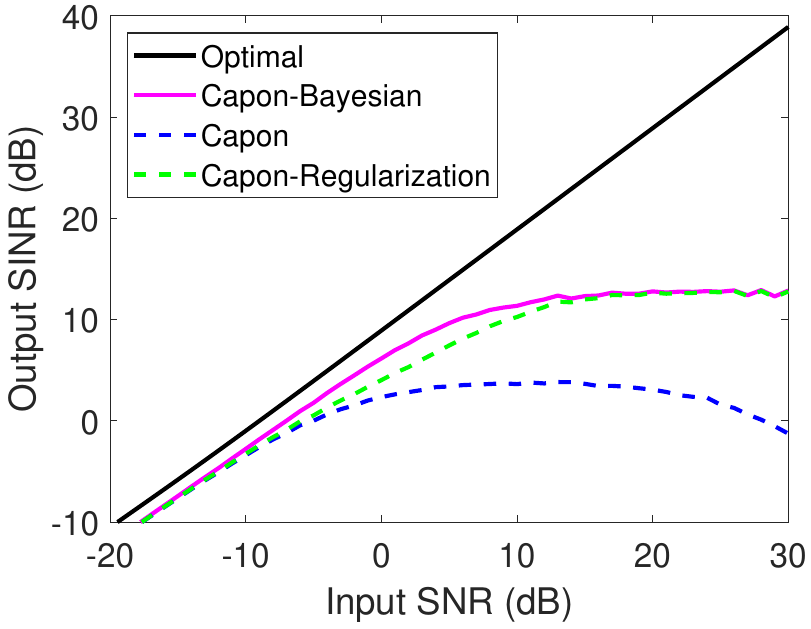}
	}
    
	\subfigure[$\beta = 0.8$]{
	 	\includegraphics[height=3cm]{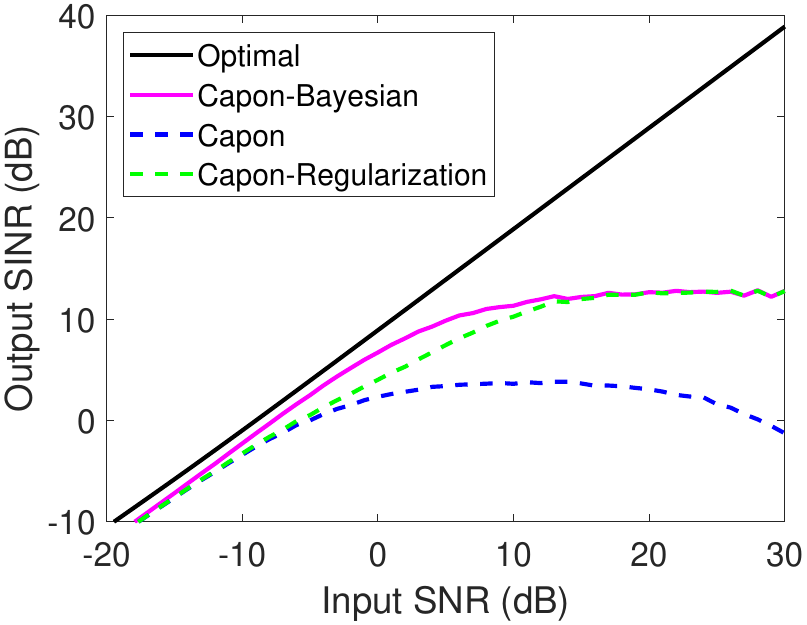}
	}
    \subfigure[$\beta = 1.0$]{
	 	\includegraphics[height=3cm]{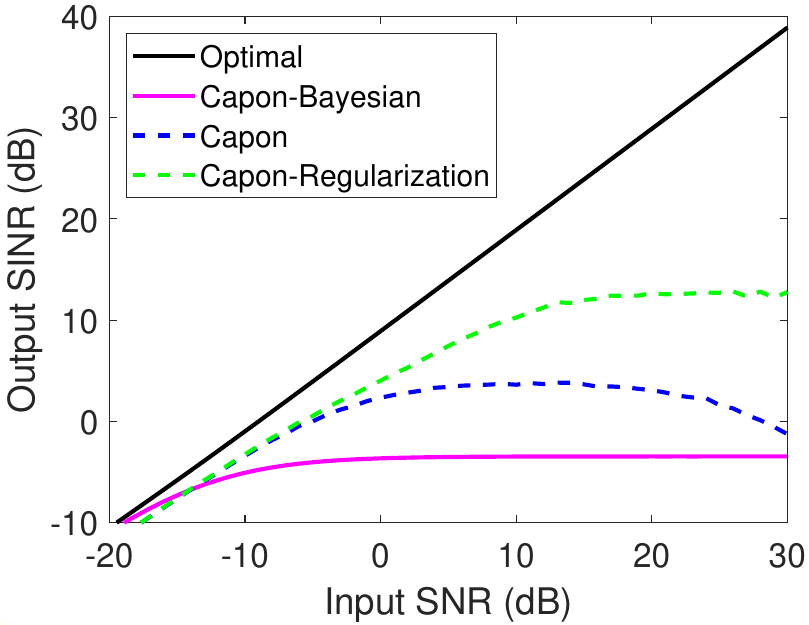}
	}
    
	\caption{Empirical verifications of the robustness principles in Insight \ref{insight:unified-framework} when no effective prior information about interferers is assumed.}
	\label{fig:robustness-sinr}
\end{figure}

Second, we suppose that the two interferers are known to be located within $[-24, -20]$ and $[28, 32]$, respectively. This prior information can be used to construct $\matb R$. To be specific, letting $\bar \Theta_1$ and $\bar \Theta_2$ be sampling subsets of $[-24, -20]$ and $[28, 32]$, respectively, with step size of $0.1$ degrees, we have
\begin{equation}\label{eq:Rbar-prior-info}
    \matb R \defeq \sum_{\theta_i \in \bar \Theta_1} \vech a(\theta_i) \vech a^\H(\theta_j) + \sum_{\theta_j \in \bar \Theta_2} \vech a(\theta_j) \vech a^\H(\theta_i) + \mat I_N.
\end{equation}
The first two terms account for interference suppression, whereas the identity matrix $\mat I_N$ provides the robustness against steering-vector uncertainties and background white noises (cf. diagonal loading). Tradeoff coefficients for the three terms can be introduced; however, for demonstration purposes, it suffices to use \eqref{eq:Rbar-prior-info} as an example. To show performance improvement when effective prior information is incorporated, in the Bayesian non-parametric principle, we let $\matb R$ be defined in \eqref{eq:Rbar-prior-info}, while in the regularization principle, we keep $\matb R = \mat I_N$ unchanged (as in Fig. \ref{fig:robustness-sinr}). The experimental results are shown in Fig. \ref{fig:robustness-sinr-prior-info}, for
which we have the following remarks: 
\begin{itemize}
    \item In the Bayesian non-parametric formulation, effective prior information indeed helps improve the ability of interference suppression.

    \item In the Bayesian non-parametric formulation, overly emphasizing the prior information (with large $\beta$), on the contrary, can deteriorate the performance because the resulting beamformer loses adaptivity to snapshots. For example, see performance degradation in Fig. \ref{fig:robustness-sinr-prior-info}(f) when input SNR is small.
\end{itemize}

\begin{figure}[!htbp]
	\centering
	\subfigure[$\beta = 0.0$]{
	 	\includegraphics[height=3cm]{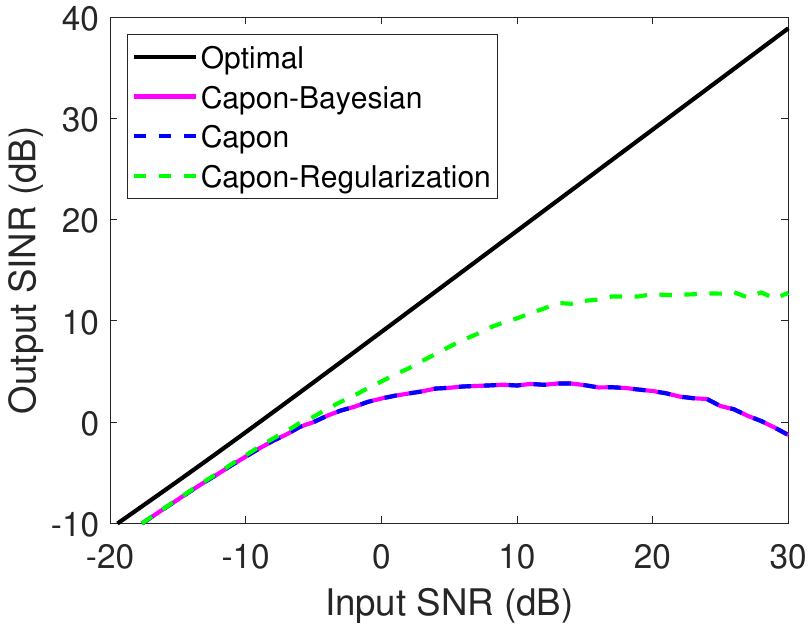}
	}
    \subfigure[$\beta = 0.2$]{
	 	\includegraphics[height=3cm]{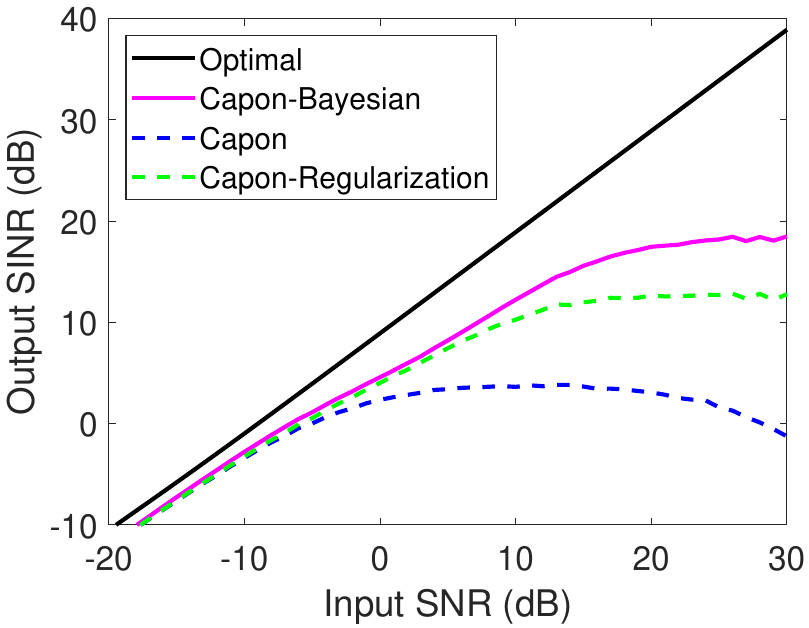}
	}

	\subfigure[$\beta = 0.4$]{
	 	\includegraphics[height=3cm]{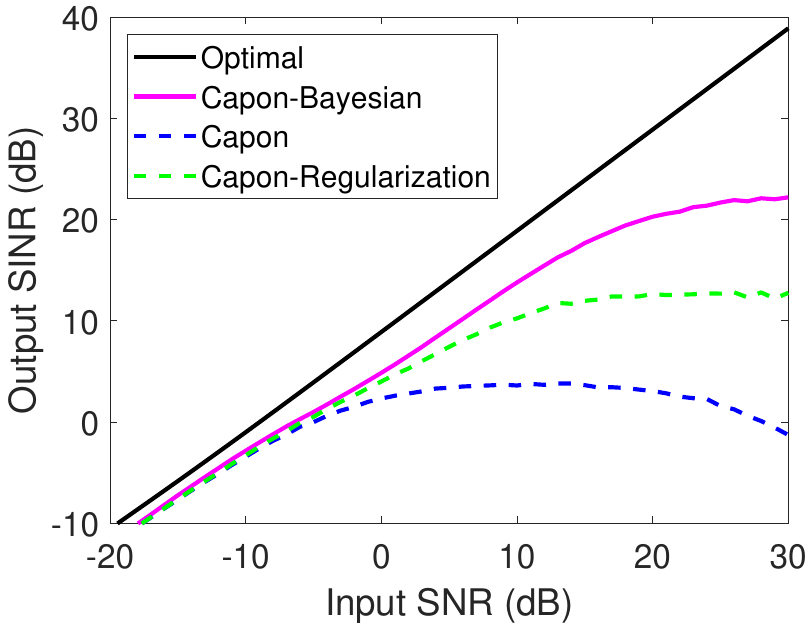}
	}
    \subfigure[$\beta = 0.6$]{
	 	\includegraphics[height=3cm]{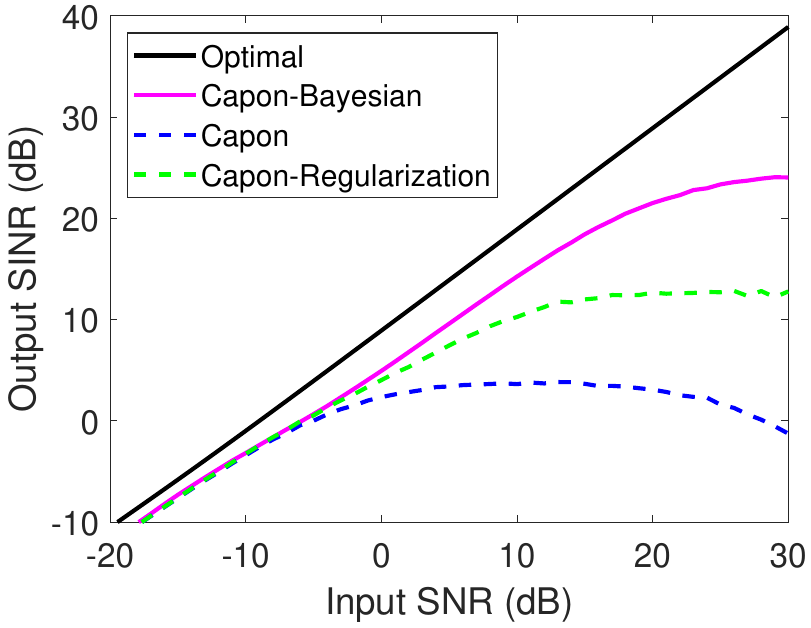}
	}
    
	\subfigure[$\beta = 0.8$]{
	 	\includegraphics[height=3cm]{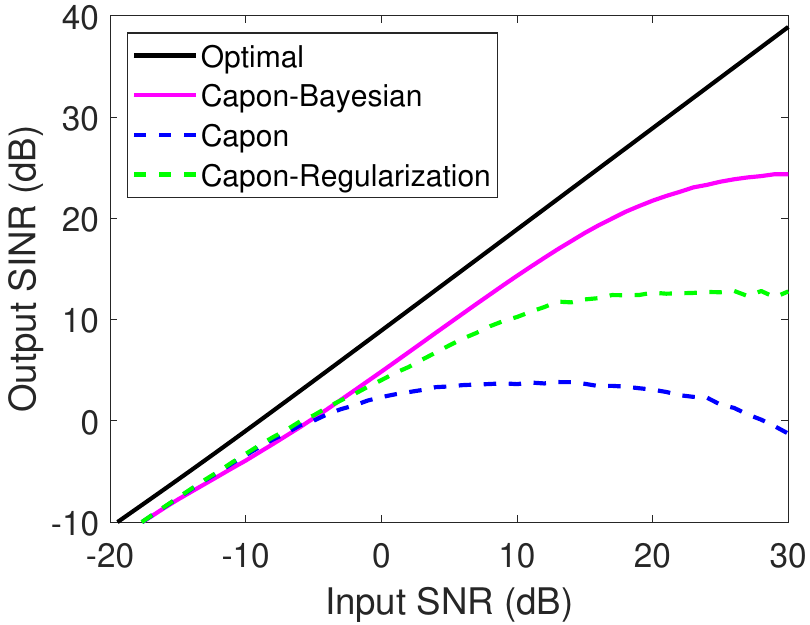}
	}
    \subfigure[$\beta = 1.0$]{
	 	\includegraphics[height=3cm]{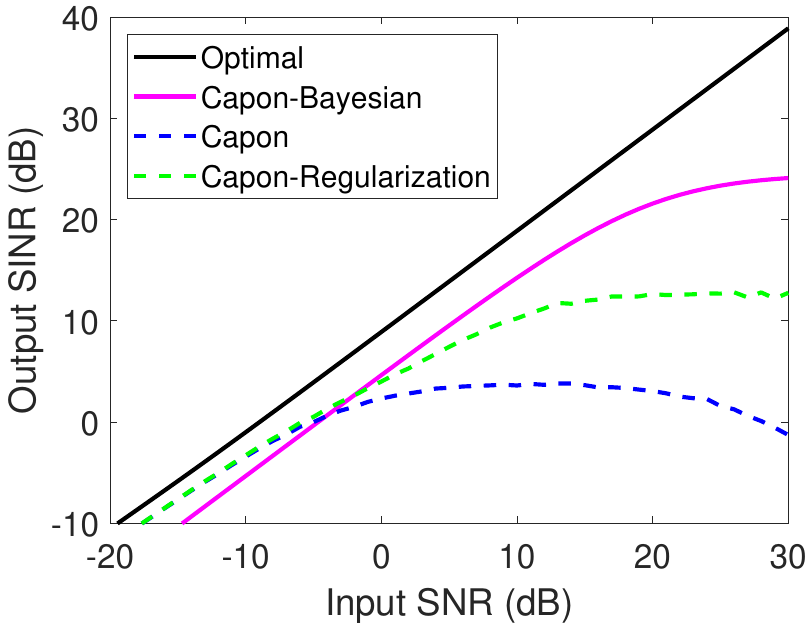}
	}
    
	\caption{Empirical verifications of the robustness principles in Insight \ref{insight:unified-framework} when effective prior location information about interferers is assumed.}
	\label{fig:robustness-sinr-prior-info}
\end{figure}

In summary, both robustness formulations in Insight \ref{insight:unified-framework} can provide significant robustness, especially when effective prior information can be given. Although the two formulations are developed from different motivations and techniques, given the same $\matb R$, they are equivalent if $\epsilon = \beta/(1-\beta)$. Due to this equivalence, we do not explicitly differentiate the two formulations in later experiments; instead, we just leverage the regularization formulation.

\subsection{Experimental Verification of Algorithm \ref{algo:bf-doa}}

In this subsection, we study the empirical performance of the proposed Algorithm \ref{algo:bf-doa}. In addition to the previously defined Optimal and Capon, the following state-of-the-art methods are implemented in the experiments for performance comparison:
\begin{itemize}
    \item IPN: Robust beamforming based on IPN covariance reconstruction \cite{gu2012robust}.
    
    \item IPN-DL: Diagonally-loaded robust beamforming based on IPN covariance reconstruction. Note that the estimation of IPN covariance can never be exact in practice. Hence, robust methods such as diagonal loading have the potential to further improve the performance of the approach in \cite{gu2012robust}. 

    \item IPN-UDL: Algorithm \ref{algo:bf-doa} of this article. IPN-DL is a special case of IPN-UDL for $\delta_1 = \delta_2 = 0$.

    \item IPN-MaxEnt: Robust beamforming based on maximum-entropy (MaxEnt) spectra estimation and IPN covariance reconstruction \cite{mohammadzadeh2020maximum}. Although the MaxEnt spectra tend to have better resolution than the Capon spectra, the biases of the MaxEnt spectra can be large \cite[Fig.~4]{schmidt1986multiple}. Hence, the performance advantage of the IPN-MaxEnt method \cite{mohammadzadeh2020maximum} over the IPN method \cite{gu2012robust} may significantly vary across different scenarios or different Monte-Carlo trials.

    \item Capon-DL: Diagonally-loaded Capon beamformer.
\end{itemize}
The first four methods are under the MVDR beamforming scheme, while the last method is under the MPDR beamforming scheme. We do not study other existing methods because they are greatly dominated by IPN \cite{gu2012robust} and IPN-MaxEnt \cite{mohammadzadeh2020maximum}. Capon and Capon-DL are kept just for the illustration of the theoretical analyses in this article.

\subsubsection{Power Spectra Estimation}

The power spectra estimated by Capon beamforming, Capon beamforming with diagonal loading, and Capon beamforming with unbalanced diagonal loading are shown in Fig. \ref{fig:capon-spectra}. Note that in evaluating the resolution of a power spectra estimation method, all emitting signals should have the same energy; that is, the interference-to-signal ratio should be 0dB \cite[p.~1020]{johnson1982application}. Fig. \ref{fig:capon-spectra} suggests that the diagonal loading operation indeed widens the beams, and hence, deteriorates the resolution. This experimental observation agrees with the theoretical analyses in the beginning of Section \ref{sec:doa-estimation}. However, the unbalanced diagonal loading trick can improve the resolution in power spectra estimation, aligning with the motivation in Insight \ref{insight:high-resolution} and Method \ref{method:robust-capon-beamforming}. Note that the larger the degree of unbalancedness is (i.e., the larger the value of $\delta_1$ when fixing $\delta_2$), the better the resolution is.

\begin{figure}[htbp]
	\centering

    \subfigure[$\delta_1 = 2$]{
	 	\includegraphics[height=3cm]{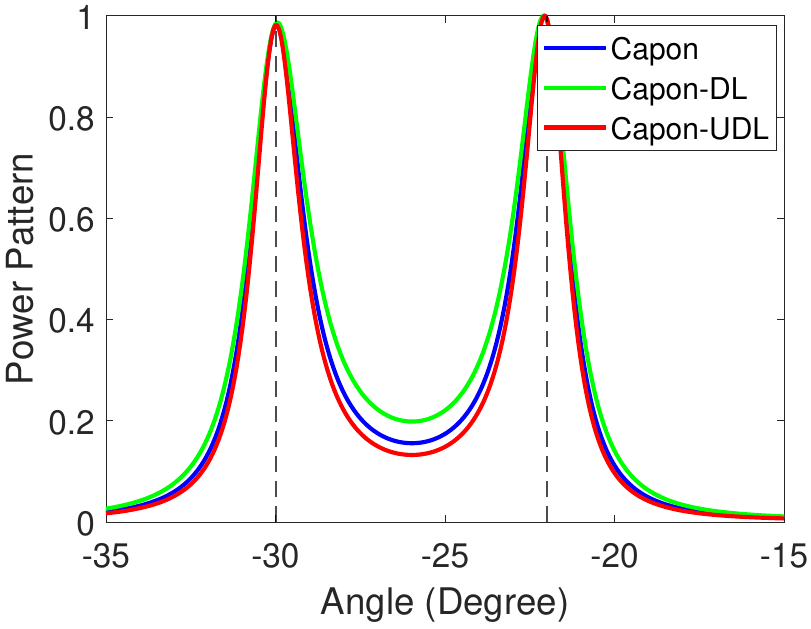}
	}
	\subfigure[$\delta_1 = 6$]{
	 	\includegraphics[height=3cm]{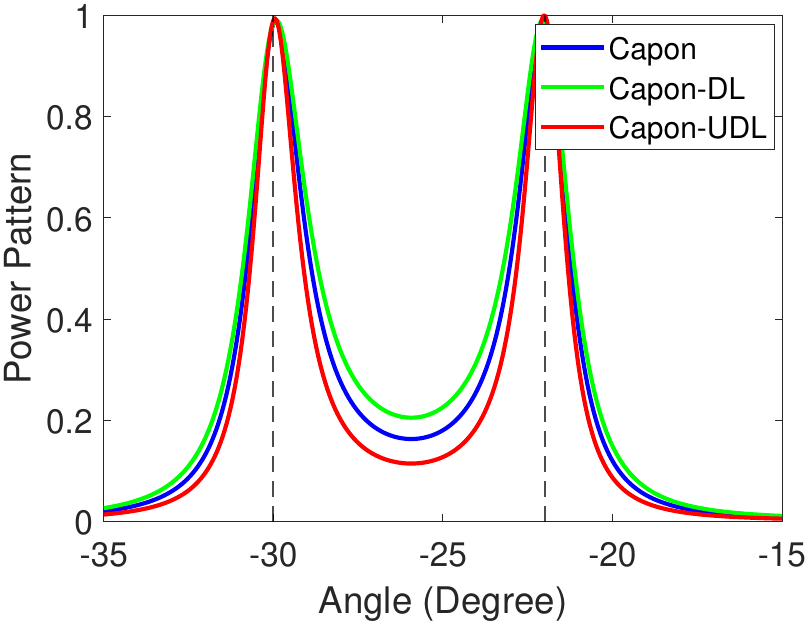}
	}
    
	\caption{Power spectra of Capon beamforming, Capon beamforming with diagonal loading (DL), and Capon beamforming with Unbalanced DL (UDL). In this scenario, ISR = 0dB and SNR = 25dB; the number of snapshots is 30; in Capon-DL, the DL parameter is $\epsilon = 0.01$; in Capon-UDL, $\delta_2 = 0.0$ (cf. Insight \ref{insight:high-resolution}).}
	\label{fig:capon-spectra}
\end{figure}

\subsubsection{Output SINR}

To compare output SINR performances, all beamforming methods are empirically tuned to achieve their best operational performances in the investigated scenarios. To be specific, in IPN-DL, the diagonal loading parameter is $\epsilon = 0.01$; in IPN-UDL, $K' = K = 3$, $\delta_1 = 10$, and $\delta_2 = \epsilon = 0.01$ (NB: large $\delta_1$ for high-resolution and nonzero $\delta_2$ for robustness; see Method \ref{method:robust-capon-beamforming}); in Capon-DL, the diagonal loading parameter is $\epsilon = 0.1$.

The output SINR performance of beamformers is shown in Fig. \ref{fig:sinr}. When plotting against the input SNR, the number of snapshots is set to 30 or 80; when plotting against the number of snapshots, the input SNR is set to 10dB or 25dB. 

\begin{figure}[!htbp]
	\centering
	\subfigure[Against SNR (Snapshot = 30)]{
	 	\includegraphics[height=3cm]{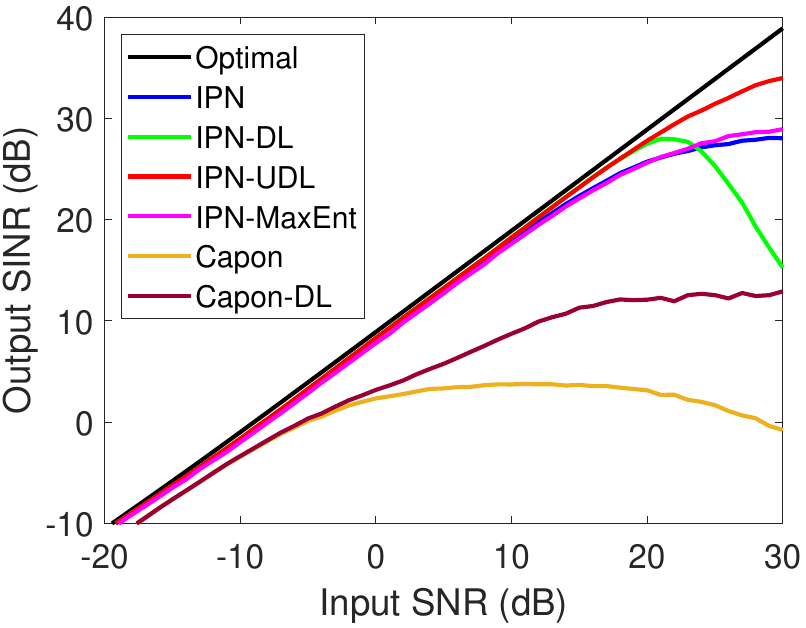}
	}
    \subfigure[Closeup of (a)]{
	 	\includegraphics[height=3cm]{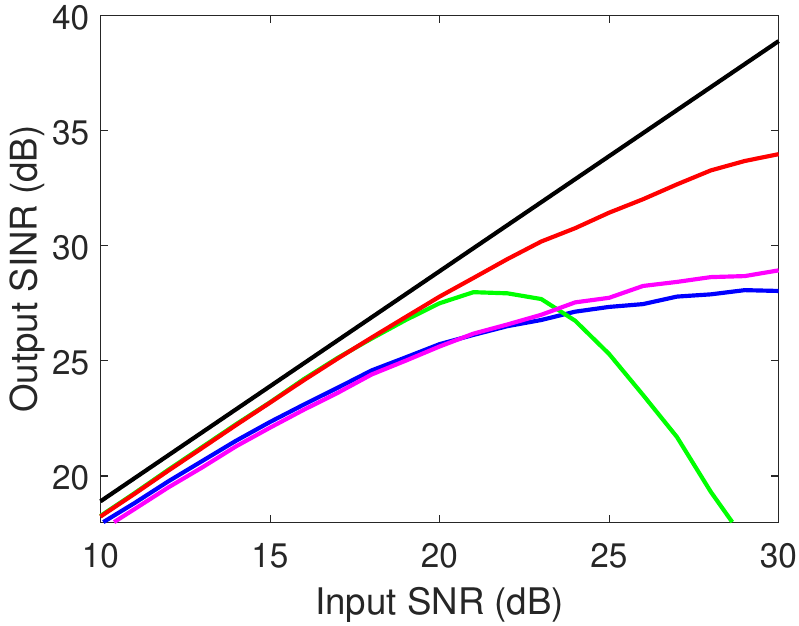}
	}

	\subfigure[Against SNR (Snapshot = 80)]{
	 	\includegraphics[height=3cm]{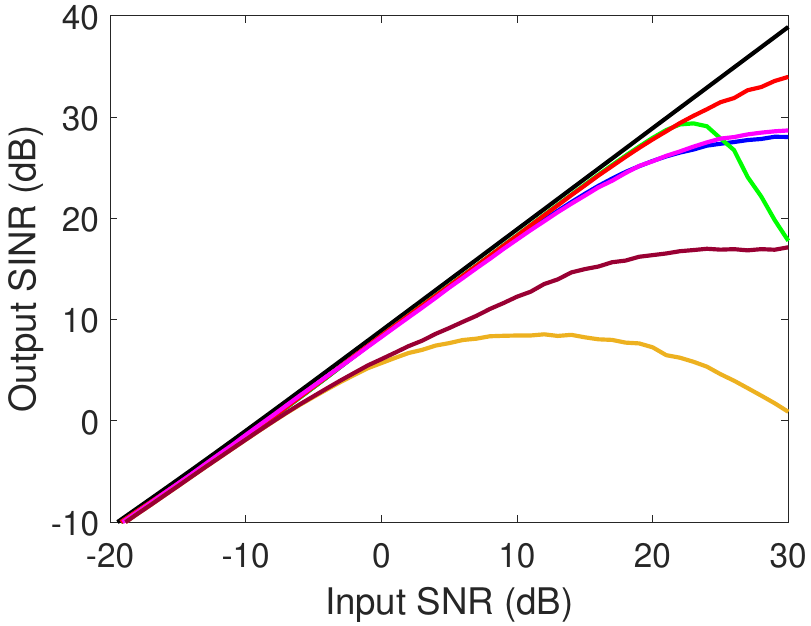}
	}
    \subfigure[Closeup of (c)]{
	 	\includegraphics[height=3cm]{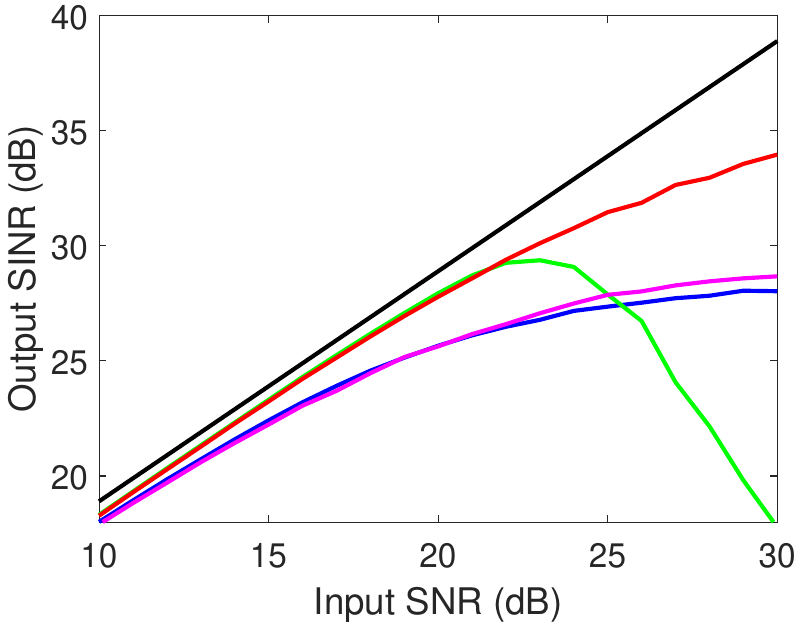}
	}
    
	\subfigure[Against Snapshot (SNR = 10dB)]{
	 	\includegraphics[height=3cm]{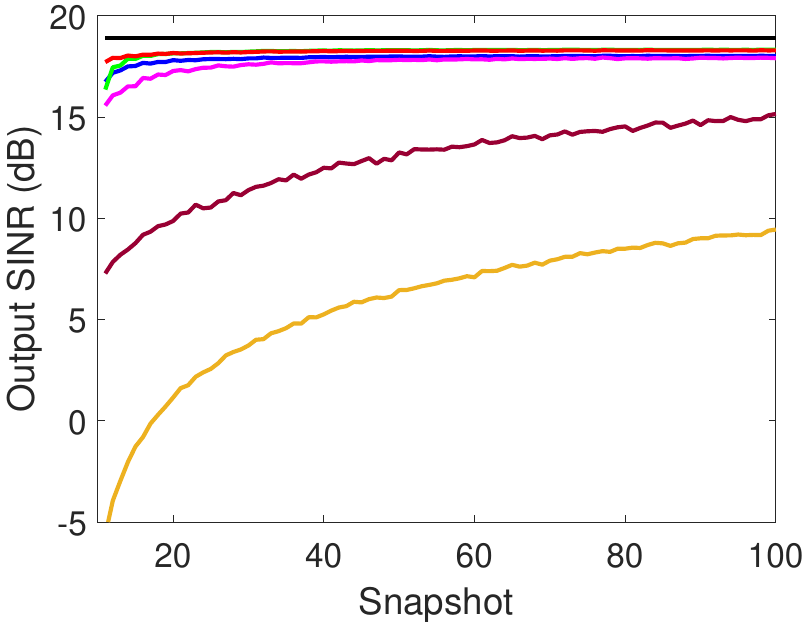}
	}
    \subfigure[Closeup of (e)]{
	 	\includegraphics[height=3cm]{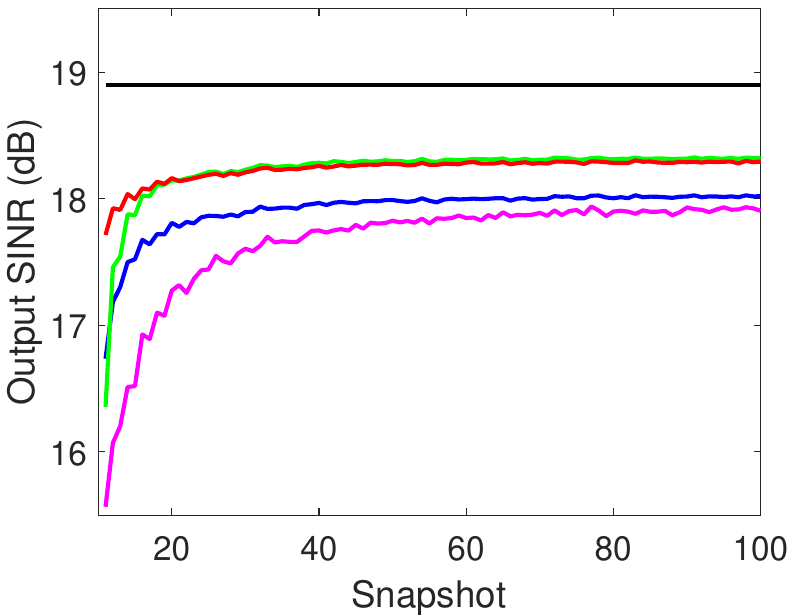}
	}

	\subfigure[Against Snapshot (SNR = 25dB)]{
	 	\includegraphics[height=3cm]{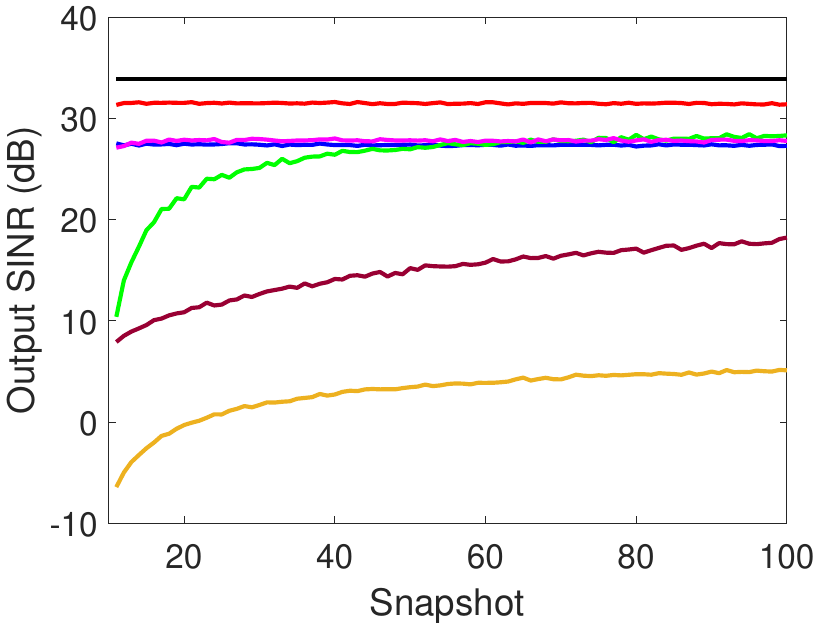}
	}
    \subfigure[Closeup of (g)]{
	 	\includegraphics[height=3cm]{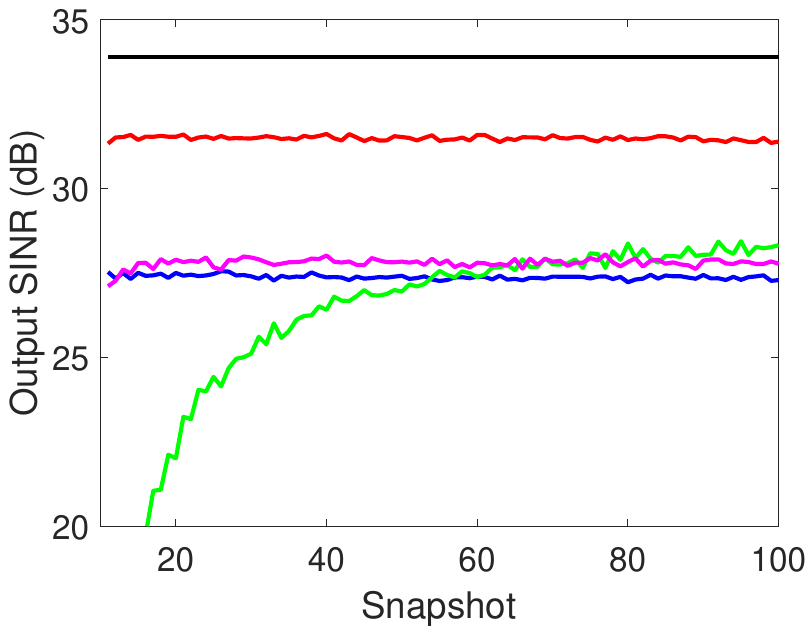}
	}
    
	\caption{The output SINR performance of beamformers against input SNR and snapshots. In (a), the number of snapshots is 30; in (c), the number of snapshots is 80; in (e), the SNR is 10dB; in (g), the SNR is 25dB.}
	\label{fig:sinr}
\end{figure}

Fig. \ref{fig:sinr} outlines the following observations.
\begin{itemize}
    \item For both IPN and Capon, diagonal loading has the potential to improve their robustness, especially in the regime of small to moderate SNRs. This is because the IPN covariance in the MVDR beamforming and the snapshot covariance in the MPDR beamforming can never be exactly estimated. Therefore, according to the analyses in Section \ref{sec:robustness-theory} (e.g., Insights \ref{insight:regularization}, \ref{insight:bayes-model}, and \ref{insight:unified-framework}), robust beamforming methods tend to outperform their nominal counterparts.

    \item IPN-UDL significantly outperforms the other beamformers due to its high-resolution and low-bias in power spectra estimation; cf. Fig. \ref{fig:capon-spectra}.
\end{itemize}

\subsubsection{Supplementary Experiments}\label{subsec:add-experiemnts}

In the experiment associated with Fig. \ref{fig:sinr}, we have assumed that $K' = K = 3$. Below, we relax this assumption and admit that $K'$ may not be exactly known. We set $K' = 5$ or $K' = 7$. The corresponding results are given in Fig. \ref{fig:sinr-k-5-7}. As we can see, the proposed IPN-UDL method is not sensitive to the value of $K'$, as long as it is reasonably set in practice. However, when $K' = 7$, the performance indeed slightly degrades compared to that for $K' = 5$ and $K' = K = 3$. This experimental observation agrees with our theoretical analyses in Subsection \ref{subsec:beamforming-high-resolution}.

\begin{figure}[htbp]
	\centering
	\subfigure[$K' = 5$]{
	 	\includegraphics[height=3cm]{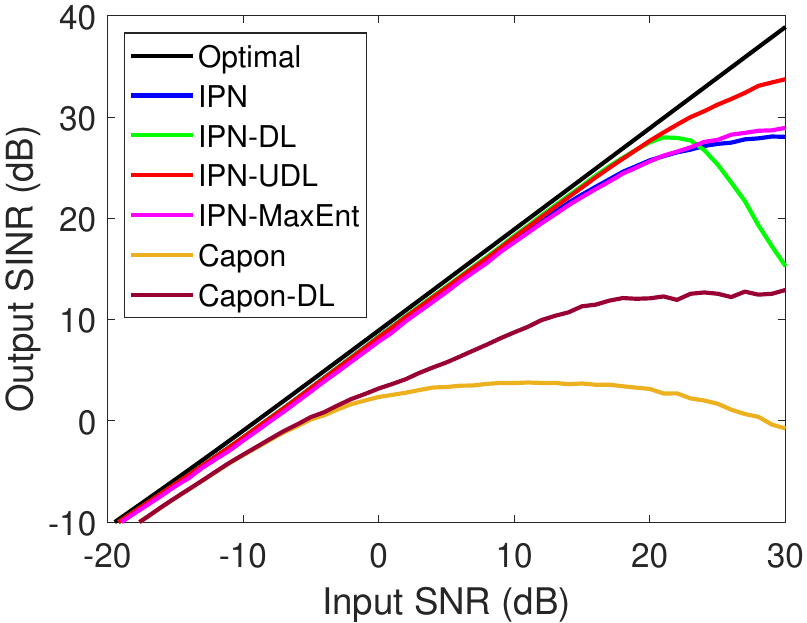}
	}
    \subfigure[$K' = 7$]{
	 	\includegraphics[height=3cm]{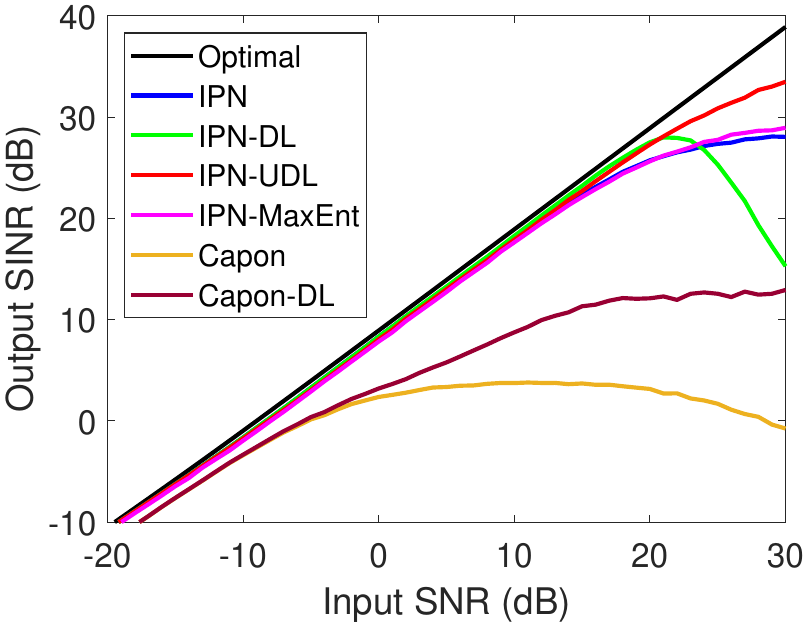}
	}

	\caption{The output SINR performance of beamformers against input SNR when $K' = 5$ and $K' = 7$. The number of snapshots is 30.}
	\label{fig:sinr-k-5-7}
\end{figure}

When the interference-to-signal ratio is 0dB and the locations of interferers are at $\theta_2 = 0^\circ$ and $\theta_3 = 30^\circ$, the output SINR performance against the input SNR is shown in Fig. \ref{fig:sinr-far-interference}. In this case, the interferers are far away from the SoI and the interference energies are equal to that of the SoI. Therefore, as reported in \cite{gu2012robust}, the IPN method is indeed a powerful tool that can achieve nearly optimal performance in both low and high SNR regimes.

\begin{figure}[htbp]
	\centering

	\subfigure[Far and Weak Interferers]{
	 	\includegraphics[height=3cm]{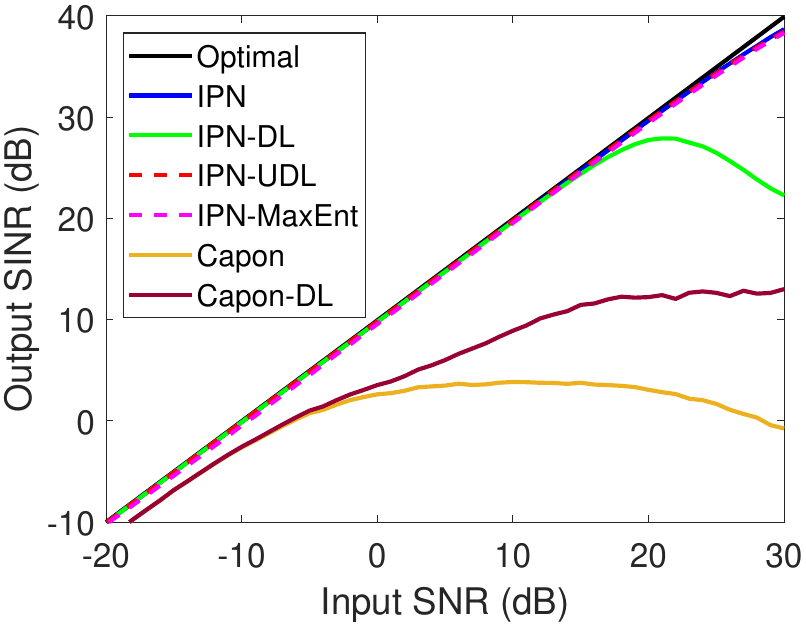}
	}
    \subfigure[Closeup of (a)]{
	 	\includegraphics[height=3cm]{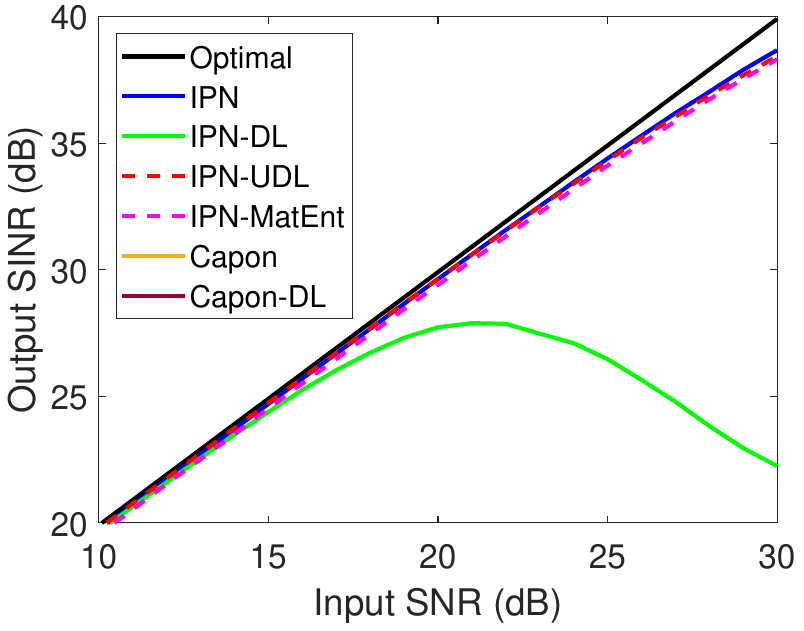}
	}

	\caption{The output SINR performance of beamformers against input SNR (far and weak interferers). The ISR is 0dB and the number of snapshots is 30.}
	\label{fig:sinr-far-interference}
\end{figure}

\section{Conclusions}\label{sec:conclusion}
To solve the problems in Subsection \ref{subsec:problem-statements} and answer the questions in Subsection \ref{subsec:research-questions}, this article comprehensively studies distributionally robust beamforming, including the conceptual system, theoretical analysis, and algorithmic design. By addressing uncertainties in snapshot and IPN covariance matrices, a unified framework for distributionally robust beamforming that generalizes several existing approaches (e.g., diagonal loading, eigenvalue thresholding, and prior-knowledge embedding) is presented (see Insight \ref{insight:unified-framework}), which draws insights from four technical approaches: locally robust, globally robust, regularized, and Bayesian-nonparametric beamforming. Our theoretical analyses demonstrate the equivalence among the four methods, and show that integrating them with subspace methods using the unbalanced diagonal-loading trick \eqref{eq:Gamma} can enhance the resolution of power spectra (and DoA) estimation. Such high resolution can improve the reconstruction accuracy of IPN covariances, and hence, improve the SINR performance when interferers are located close to the SoI.

\appendices

\section{Proof of Theorem \captext{\ref{thm:robust-bf-R0}}}\label{append:robust-bf-R0}
\begin{proof}
First, Problem \eqref{eq:robust-bf-R0} is equivalent to
\[
    \begin{array}{cl}
       \displaystyle \min_{\vec w \in \cal W} \min_{k} & k \\
        \st & \displaystyle \max_{\mat R \in \cal B_{\epsilon}(\mat R_0)} \vec w^\H \mat R \vec w  - \vec w_0^\H  \mat R_0 \vec w_0 \le k, \\
        & k \ge 0.
    \end{array}
\]
By eliminating $k$, the above display is equivalent to
\[
    \displaystyle \min_{\vec w \in \cal W} \max_{\mat R \in \cal B_{\epsilon}(\mat R_0)}  \vec w^\H \mat R \vec w   - \vec w_0^\H  \mat R_0 \vec w_0,
\]
which is further equivalent to \eqref{eq:minimax-robust-bf-R0}. Note that $\max_{\mat R}  \vec w^\H \mat R \vec w   - \vec w_0^\H  \mat R_0 \vec w_0 \ge 0$ because $\mat R_0 \in \cal B_{\epsilon}(\mat R_0)$. The second statement can be similarly proven.
\stp
\end{proof}

\section{Proof of Theorem \captext{\ref{thm:sol-robust-bf-Rh-global}}}\label{append:sol-robust-bf-Rh-global}
\begin{proof}
If $d(\mat R, \math R) \defeq \Tr[\mat R - \math R]^\H[\mat R - \math R]$, Problem \eqref{eq:robust-bf-Rh-global} can be written as
\[
    \begin{array}{cl}
       \displaystyle \min_{\vec w \in \cal W, k}  & k \\
        \st & \displaystyle \max_{\mat R \in \C^{N \times N}} \vec w^\H \mat R \vec w  - k \Tr[\mat R - \math R]^\H[\mat R - \math R] \le \tau. \\
        & k \ge 0.
    \end{array}
\]
In terms of $\mat R$, this is a quadratic convex and unconstrained optimization, and therefore, the first-order optimality condition $2k(\mat R - \math R) = \vec w \vec w^\H$ gives the globally optimal solution $\mat R = \math R + \vec w \vec w^\H / (2k)$. Hence, the above optimization is equivalent to
\[
    \begin{array}{cl}
       \displaystyle \min_{\vec w \in \cal W, k}  & k \\
        \st & \vec w^\H \left[\math R + \displaystyle \frac{\vec w \vec w^\H}{4k} \right] \vec w \le \tau \\
        & k \ge 0.
    \end{array}
\]
Because $\vec w^\H \left[\math R + \frac{\vec w \vec w^\H}{4k} \right] \vec w$ is continuous and monotonically decreasing in $k$ for every $\vec w \in \cal W$, the above optimization is further equivalent to 
\[
    \begin{array}{cl}
       \displaystyle \min_{k}  & k \\
        \st & \displaystyle \min_{\vec w \in \cal W} \vec w^\H \left[\math R + \displaystyle \frac{\vec w \vec w^\H}{4k} \right] \vec w \le \tau \\
        & k \ge 0.
    \end{array}
\]
Since $\min_{\vec w \in \cal W} \vec w^\H \left[\math R + \frac{\vec w \vec w^\H}{4k} \right] \vec w$ is also continuous and monotonically decreasing in $k$ and it tends to infinity when $k \to 0$, the above optimization is equivalent to 
\[
    \begin{array}{cl}
       \displaystyle \operatorname{find}   &  k \\
        \st & \displaystyle \min_{\vec w \in \cal W} \vec w^\H \left[\math R + \displaystyle \frac{\vec w \vec w^\H}{4k} \right] \vec w = \tau \\
        & k \ge 0.
    \end{array}
\]
This proves the first part of the theorem. The second part is technically similar; just note that 
\[
\begin{array}{cl}
d(\mat R, \math R) & \defeq \vec w^\H \Tr[\mat R - \math R]^\H[\mat R - \math R] \vec w \\
&= \vec w^\H \vec w \Tr[\mat R - \math R]^\H[\mat R - \math R],
\end{array}
\]
and therefore, the first-order optimality condition $2k \vec w^\H \vec w (\mat R - \math R) = \vec w \vec w^\H$ gives the globally optimal solution 
\[
\mat R = \math R + \vec w \vec w^\H / (2k \vec w^\H \vec w).
\]
The third part is also technically similar; just note that 
\[
\begin{array}{cl}
d(\mat R, \math R) & \defeq \vec w^\H \Tr[\mat R - \math R]^\H \mat C^{-1} [\mat R - \math R] \vec w \\
&= \vec w^\H \vec w \Tr[\mat R - \math R]^\H \mat C^{-1} [\mat R - \math R],
\end{array}
\]
and therefore, the first-order optimality condition $2k \vec w^\H \vec w \mat C^{-1} (\mat R - \math R) = \vec w \vec w^\H$ gives the globally optimal solution 
\[
\mat R = \math R + \mat C \vec w \vec w^\H / (2k \vec w^\H \vec w).
\]
This completes the proof.
\stp
\end{proof}

\bibliographystyle{IEEEtran}
\bibliography{References}

\begin{IEEEbiography}[{\includegraphics[width=1in,height=1.25in,clip,keepaspectratio]{wsx}}]{Shixiong Wang} received the B.Eng. degree in detection, guidance, and control technology, and the M.Eng. degree in systems and control engineering from the School of Electronics and Information, Northwestern Polytechnical University, China, in 2016 and 2018, respectively. He received his Ph.D. degree from the Department of Industrial Systems Engineering and Management, National University of Singapore, Singapore, in 2022. He is currently a Postdoctoral Research Associate with the Intelligent Transmission and Processing Laboratory, Imperial College London, London, United Kingdom, from May 2023. He was a Postdoctoral Research Fellow with the Institute of Data Science, National University of Singapore, Singapore, from March 2022 to March 2023. His research interest includes statistics and optimization theories with applications in signal processing (e.g., optimal estimation theory) and machine learning (e.g., generalization error theory).
\end{IEEEbiography}

\begin{IEEEbiography}
[{\includegraphics[width=1in,height=1.2in,clip,keepaspectratio]{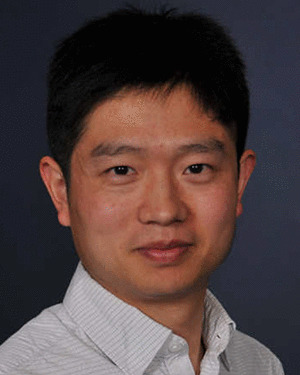}}] 
{Wei Dai} (Member, IEEE) received the Ph.D. degree from the University of Colorado Boulder, Boulder, Colorado, in 2007. He is currently a Senior Lecturer (Associate Professor) in the Department of Electrical and Electronic Engineering, Imperial College London, London, UK. From 2007 to 2011, he was a Postdoctoral Research Associate with the University of Illinois Urbana-Champaign, Champaign, IL, USA. His research interests include electromagnetic sensing, biomedical imaging, wireless communications, and information theory.
\end{IEEEbiography}

\begin{IEEEbiography}
[{\includegraphics[width=1in,height=1.2in,clip,keepaspectratio]{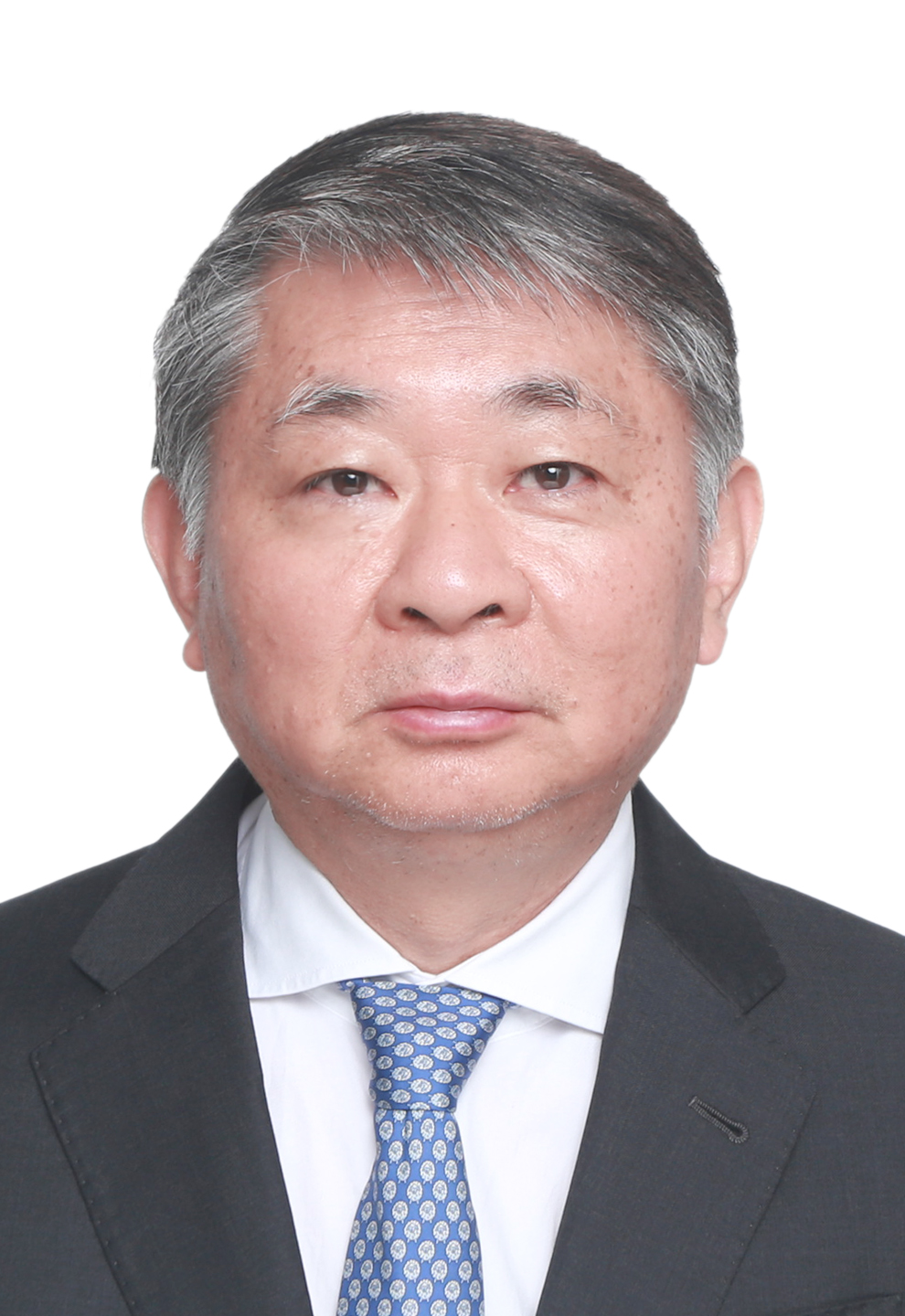}}] 
{Geoffrey Ye Li} (Fellow, IEEE) is currently a Chair Professor at Imperial College London, UK.  Before joining Imperial in 2020, he was a Professor at Georgia Institute of Technology for 20 years and a Principal Technical Staff Member with AT\&T Labs – Research (previous Bell Labs) for five years. He made fundamental contributions to orthogonal frequency division multiplexing (OFDM) for wireless communications, established a framework on resource cooperation in wireless networks, and introduced deep learning to communications. In these areas, he has published over 700 journal and conference papers in addition to over 40 granted patents. His publications have been cited around 80,000 times with an H-index over 130. He has been listed as a Highly Cited Researcher by Clarivate/Web of Science almost every year. Dr. Geoffrey Ye Li was elected to Fellow of the Royal Academic of Engineering (FREng), IEEE Fellow, and IET Fellow for his contributions to signal processing for wireless communications. He received 2024 IEEE Eric E. Sumner Award, 2019 IEEE ComSoc Edwin Howard Armstrong Achievement Award, and several other awards from IEEE Signal Processing, Vehicular Technology, and Communications Societies.

\end{IEEEbiography}

\end{document}